\documentclass[%
reprint,
superscriptaddress,
 amsmath,amssymb,
 aps,
 prx
prl
]{revtex4-2}

\usepackage[T1]{fontenc}
\usepackage{graphicx}
\usepackage{dcolumn}
\usepackage{bm}
\usepackage[colorlinks=true, allcolors=blue]{hyperref}
\usepackage{braket}
\usepackage[table]{xcolor}
\usepackage{xfrac}
\usepackage{comment}
\usepackage[normalem]{ulem}
\usepackage{mathtools}
\usepackage{layouts}
\usepackage{quantikz}
\usepackage{tabularx}
\usepackage{multirow}
\usepackage{esint}
\usepackage{bbm}
\usepackage{mathdots}
\usepackage{soul}
\usetikzlibrary{shapes.geometric, arrows}
\usepackage{adjustbox}

\usepackage{amsthm}
\usepackage[capitalise]{cleveref}

\newcommand{\approptoinn}[2]{\mathrel{\vcenter{
  \offinterlineskip\halign{\hfil$##$\cr
    #1\propto\cr\noalign{\kern2pt}#1\sim\cr\noalign{\kern-2pt}}}}}

\begin{document}
\title{Solving Free Fermion Problems on a Quantum Computer}
\author{Maarten Stroeks}
\email[]{m.e.h.m.stroeks@tudelft.nl}
\affiliation{QuTech, TU Delft,  Lorentzweg 1, 2628 CJ Delft, The Netherlands}
\affiliation{Delft Institute of Applied Mathematics, TU Delft,  2628 CD Delft,  The Netherlands}
\author{Daan Lenterman}
\affiliation{Department of Physics, ETH Z\"urich, CH-8093 Z\"urich, Switzerland}
\author{Barbara M. Terhal}
\affiliation{QuTech, TU Delft,  Lorentzweg 1, 2628 CJ Delft, The Netherlands}
\affiliation{Delft Institute of Applied Mathematics, TU Delft,  2628 CD Delft,  The Netherlands}
\author{Yaroslav Herasymenko}
\email[]{yaroslav@cwi.nl}
\affiliation{Delft Institute of Applied Mathematics, TU Delft,  2628 CD Delft,  The Netherlands}
\affiliation{QuSoft and CWI, Science Park 123,  1098 XG Amsterdam, The Netherlands}
\newtheorem{theorem}{Theorem}
\newtheorem{corollary}[theorem]{Corollary}
\newtheorem{lemma}[theorem]{Lemma}
\newtheorem{proposition}[theorem]{Proposition}
\newtheorem{claim}[theorem]{Claim (informal)}
\newtheorem{problem}[theorem]{Problem}
\newtheorem{criterion}{Criterion}
\newtheorem{appendixcriterion}{Criterion}[section]
\theoremstyle{definition}
\newtheorem{assumption}{Assumption}
\theoremstyle{remark}
\newtheorem*{remark}{Remark}
\theoremstyle{definition}
\newtheorem{definition}[theorem]{Definition}
\Crefname{criterion}{Criterion}{Criteria}
\crefname{criterion}{Crit.}{Crit.}

\begin{abstract}
Simulating noninteracting fermion systems is a common task in computational many-body physics. In absence of translational symmetries, modeling free fermions on $N$ modes usually requires ${\rm poly}(N)$ computational resources. While often moderate, these costs can be prohibitive in practice when large systems are considered.
\textcolor{black}{We present several free-fermion problems that can be solved by a quantum algorithm with substantially reduced computational costs. The memory costs are exponentially improved, ${\rm poly}\,{\rm log}(N)$. The runtime improvement, compared to the best known classical algorithms, is either exponential or significantly polynomial, depending on the geometry of the problem. The simulation of free-fermion dynamics belongs to the BQP-hard complexity class. This implies (under standard assumptions) that our algorithm yields an exponential speedup for \textit{any} classical algorithm at least for some geometries.} The key technique in our algorithm is the block-encoding of objects such as correlation matrices and Green's functions into a unitary. We demonstrate how such unitaries can be efficiently realized as quantum circuits, in the context of dynamics and thermal states of tight-binding Hamiltonians. The special cases of disordered and inhomogeneous lattices, as well as large non-lattice graphs, are presented in detail. Finally, we show that our simulation algorithm generalizes to other promising targets, including free boson systems.
\end{abstract}
\maketitle

\clearpage 

\section{Introduction and background} Quantum many-body dynamics can be naturally simulated by a quantum computer \cite{LLoyd1996}, enabling its applications in condensed matter and quantum chemistry.
For a system of size $N$, standard quantum algorithms use ${\rm poly}(N)$ resources for such simulations. It implies an exponential advantage over classical methods, when dealing with a generic many-body system. Such a general advantage may not hold in special cases of interest, such as the modeling of free fermions, where the best classical algorithms also have ${\rm poly}(N)$ cost \cite{TD:freefermion,knill2001,bravyi:FLO}. This classical efficiency has been key to many successes of computational physics, as free fermions model a variety of systems in condensed matter and quantum chemistry; they have also been used in computational strategies for solving interacting fermion systems, using mean-field (Hartree-Fock), perturbative methods or dynamical mean-field theory. Nonetheless, in the practical simulations of noninteracting fermions, even the most efficient numerical methods become too expensive for large systems. This motivates the key question of this work: can a quantum computer boost free-fermion simulations beyond what can be done classically? We answer this question in the affirmative, presenting quantum algorithms with an exponential speedup and memory compression for several free-fermion problems.

To appreciate the value of such an exponential reduction, consider numerical simulations of free-fermion models of materials and interfaces for quantum transport \cite{Groth_2014, Kloss_2021}. These can become prohibitive when involving more than $N=10^9$ modes, which is of practical interest when simulating $3$-dimensional lattice models. Upon compression, a system of $10^9$ sites can be described by $n=30$ qubits. Larger systems of practical interest could still be accessed with moderately sized quantum computers. Indeed, even simulating one mole ($N\simeq 10^{24}$) of fermionic modes requires fewer than $n=80$ qubits in compressed form. This opens the door to modeling free fermions near the thermodynamic limit --- a desired but often challenging goal.

Our result is based on an understanding of the reduced classical complexity of free-fermion systems. As an inspiration, we used the fact that the matchgate computations and the dynamics of free fermion problems on $N=2^n$ modes 
can be simulated in compressed form, using $O(n)$ \textit{space} on a quantum computer \cite{JKMW,Kraus_2011, supremeFF,senjean2023toward,GBC}. In this work, we go beyond these memory compression results to identify free fermion problems that also permit an exponentially improved, ${\rm poly} (n)$ quantum \textit{run-time}. Our key idea is to represent the relevant $2^n$-sized object---such as the correlation matrix or a Green's function of a free-fermion state---as a block of an $n$-qubit unitary. This unitary can be given as an efficient quantum circuit; we provide explicit construction methods by leveraging the modern quantum algorithm toolbox of block-encoding manipulations \cite{berry2015hamiltonian, Low_2019, Gilyen_2019, lin2022lecturenotesquantumalgorithms, Rall_2020}. In particular, we show how to construct the desired unitary for free-fermion states arising from time dynamics or thermal equilibrium. Given block-encodings of the aforementioned objects into a circuit, we show how to accurately extract various physical quantities for a state, including the occupation number on a given site, or energy density across the entire system. We analyze the application of our methods to free-fermion models on $d$-dimensional lattices and expander graphs. For the particular case of $d$-dimensional lattices, we argue that a polynomial runtime improvement can be expected, based on the comparison with best available classical algorithms. For expander graphs, the same analysis suggests an \textit{exponential} speedup.
On a more general geometry, the problem of single-particle time dynamics is BQP-hard \cite{babbush2023exponential} --- as hard as any problem that can be efficiently solved by a quantum computer. This rigorously proves that our approach offers an exponential quantum speedup at least for some geometries (as long as quantum computers can offer exponential speedups in principle). Finally, we outline the generalization of our approach to systems beyond free fermions.

Our work can be viewed as a fermionic counterpart to \cite{babbush2023exponential}, which shows how the time-dynamics of a system of coupled oscillators can be solved exponentially faster on a quantum versus a classical computer --- with further applications in \cite{danzetal:oscillators}. Compared to the alternative and recent work \cite{GoogleShadow} which focuses on encoding a correlation matrix into a state, our method using block-encodings has an exponential advantage in signal strength for the extraction of local observables (see Appendix \ref{sec:alt_encodings} for more details). 

We note that quantum algorithms for compressed simulation of \textit{interacting} fermionic models have also been considered in e.g. \cite{StanisicMontanaro}, where a Fermi-Hubbard model is simulated in the $O(1)$-particle sub-space. By contrast, we consider $N$-mode systems with as many as $\Theta(N)$ particles.

\section{Preliminaries} Throughout this work, we set $N=2^n$. A particle-conserving free fermion Hamiltonian $H$ can be written as
\begin{equation}
    H=\sum_{i=0,j=0}^{N-1,N-1}h_{ij} a_j^{\dagger}a_i,
    \label{eq:defhmain}
\end{equation}
with Hermitian matrix $h$, which we will assume to be $O(1)$-sparse (i.e., there are at most a constant number of non-zero entries in each row) and $|h_{ij}|\leq 1$. Here $\{a_i^{\dagger},a_j\}=\delta_{ij}, \{a_i,a_j\}=\{a_i^{\dagger},a_j^{\dagger}\}=0$. We denote the fermionic particle number operator as $\hat{N}=\sum_{i=0}^{N-1} a_i^{\dagger}a_i$, and we restrict ourselves to Hamiltonians which preserve particle number \footnote{There are straightforward generalizations, using Majorana fermion language, to just parity-conserving free fermion Hamiltonians.}. We allow for states $\rho$ with an arbitrary number of particles ${\rm Tr}\, (\hat{N} \rho)$, which in general may scale with $N=2^n$. Observe that in the case of single-particle dynamics ${\rm Tr}\, (\hat{N} \rho)=1$, the fermionic nature of the system does not come into play and bosonic or fermionic dynamics are equivalent.

The Hermitian correlation matrix $M$ of a fermionic state $\rho$ on $N$ modes is defined as
\begin{equation}
   M_{ij}={\rm Tr}\, (a_i^{\dagger} a_j \rho) \in \mathbb{C}, 
   \label{eq:defMmain}
\end{equation}
and obeys $0 \leq M \leq I$, and ${\rm Tr}(M)=\langle \hat{N}\rangle$. $M$ contains observable information about the fermionic state $\rho$: for example, $M_{jj}$ is the mean fermion occupation number of a state $\rho$ in the mode $j$. Furthermore, an expectation value of a free fermion Hamiltonian (Eq.\,\eqref{eq:defhmain}) can be expressed as $\mathrm{Tr}\left(H\rho  \right)=\sum_{i,j}h_{ij}M_{ji}$. If $\rho$ is itself free-fermionic, expectation values of \textit{interacting} Hamiltonians can also be obtained from $M$, using Wick's theorem.  

Throughout this work, we will use $[N=2^n]$ in a non-traditional way, namely offset by 1: $[N]\equiv\{0,\ldots, N-1\}$. We also use the standard notation $f(x)=O(g(x))$ if a function is asymptotically upper bounded by $\mathrm{const}\cdot g(x)$, $f(x)=\Omega(g(x))$ if lower bounded, and $f(x)=\Theta(g(x))$ if both (i.e., scaling in the same way as $\mathrm{const}\cdot g(x)$).


\definecolor{e47fa927-aed0-5033-ae17-129d4d15d95c}{RGB}{241, 240, 233}
\definecolor{f3551e38-74df-57e2-b793-83d7fe876c85}{RGB}{0, 0, 0}
\definecolor{0b71a967-1f15-55a5-9bb9-70efa7b4fc58}{RGB}{51, 51, 51}
\tikzstyle{613e886f-9a38-5dd0-b8dc-e63bf7dad410} = [rectangle, rounded corners, minimum width=3cm, minimum height=1cm, text centered, font=\normalsize, color=0b71a967-1f15-55a5-9bb9-70efa7b4fc58, draw=f3551e38-74df-57e2-b793-83d7fe876c85, line width=1, fill=e47fa927-aed0-5033-ae17-129d4d15d95c]
\tikzstyle{2f8566b6-7e2c-5bfe-a65a-4b950ff4847a} = [rectangle, rounded corners, minimum width=3cm, minimum height=1cm, text centered, font=\normalsize, color=f3551e38-74df-57e2-b793-83d7fe876c85, draw=f3551e38-74df-57e2-b793-83d7fe876c85, line width=1, fill=e47fa927-aed0-5033-ae17-129d4d15d95c]
\tikzstyle{9617732b-28d6-5584-babe-0f80052099e0} = [rectangle, rounded corners, minimum width=3cm, minimum height=1cm, text centered, font=\normalsize, color=0b71a967-1f15-55a5-9bb9-70efa7b4fc58, draw=f3551e38-74df-57e2-b793-83d7fe876c85, line width=1, fill=e47fa927-aed0-5033-ae17-129d4d15d95c]
\tikzstyle{b0874d4a-ee06-54b9-a83e-c54ec9b1ace2} = [rectangle, rounded corners, minimum width=3cm, minimum height=1cm, text centered, font=\normalsize, color=0b71a967-1f15-55a5-9bb9-70efa7b4fc58, draw=f3551e38-74df-57e2-b793-83d7fe876c85, line width=1, fill=e47fa927-aed0-5033-ae17-129d4d15d95c]
\tikzstyle{caf4fe8a-33dc-547f-a2b9-0c6fc7af7cb8} = [rectangle, rounded corners, minimum width=4cm, minimum height=1cm, text centered, font=\normalsize, color=0b71a967-1f15-55a5-9bb9-70efa7b4fc58, draw=f3551e38-74df-57e2-b793-83d7fe876c85, line width=1, fill=e47fa927-aed0-5033-ae17-129d4d15d95c]
\tikzstyle{0205d17f-fb08-58b1-bb55-01711fd4389d} = [thick, draw=f3551e38-74df-57e2-b793-83d7fe876c85, line width=2, ->, >=stealth]
\tikzstyle{d025e20a-e5fd-5f48-9ef5-0b5d82a70869} = [thick, draw=f3551e38-74df-57e2-b793-83d7fe876c85, line width=2, ->, >=stealth]

\begin{figure*}[t]
\centering
\begin{adjustbox}{width=0.85\textwidth}
\begin{tikzpicture}[node distance=2cm]
\node (2cb74b6a-1132-4912-8a21-f2a2f875939d) [613e886f-9a38-5dd0-b8dc-e63bf7dad410] {Sparse access to h};
\node (d62dbbd2-af00-45fe-b12f-5ab04e0ee4f2) [2f8566b6-7e2c-5bfe-a65a-4b950ff4847a, right of=2cb74b6a-1132-4912-8a21-f2a2f875939d, xshift=2.2cm] {Block-encoding of $h$};
\node (0d35964c-0680-4a2d-bfb6-ec9d4cde3071) [9617732b-28d6-5584-babe-0f80052099e0, right of=d62dbbd2-af00-45fe-b12f-5ab04e0ee4f2, xshift=2.5cm] {Block-encoding of $F(h)$};
\node (fecc014b-c71b-4bfc-9f55-7997e5084e6f) [b0874d4a-ee06-54b9-a83e-c54ec9b1ace2, right of=0d35964c-0680-4a2d-bfb6-ec9d4cde3071, xshift=3.0cm] {Estimate entries of $U_{F(h)}$};
\node (eaa49b24-6e3a-46b3-a0dc-ed1b1e701770) [caf4fe8a-33dc-547f-a2b9-0c6fc7af7cb8, below of=d62dbbd2-af00-45fe-b12f-5ab04e0ee4f2] {Polynomial approximation of $F(x)$};
\node (48c97e8d-3524-4c99-bdc7-632b332e1fd3) [613e886f-9a38-5dd0-b8dc-e63bf7dad410, right of=eaa49b24-6e3a-46b3-a0dc-ed1b1e701770, xshift=4.5cm] {Hadamard test};
\draw [0205d17f-fb08-58b1-bb55-01711fd4389d] (2cb74b6a-1132-4912-8a21-f2a2f875939d) --  (d62dbbd2-af00-45fe-b12f-5ab04e0ee4f2);
\draw [d025e20a-e5fd-5f48-9ef5-0b5d82a70869] (d62dbbd2-af00-45fe-b12f-5ab04e0ee4f2) --  (0d35964c-0680-4a2d-bfb6-ec9d4cde3071);
\draw [d025e20a-e5fd-5f48-9ef5-0b5d82a70869] (0d35964c-0680-4a2d-bfb6-ec9d4cde3071) --  (fecc014b-c71b-4bfc-9f55-7997e5084e6f);
\draw [d025e20a-e5fd-5f48-9ef5-0b5d82a70869] (eaa49b24-6e3a-46b3-a0dc-ed1b1e701770) -|  (0d35964c-0680-4a2d-bfb6-ec9d4cde3071);
\draw [d025e20a-e5fd-5f48-9ef5-0b5d82a70869] (48c97e8d-3524-4c99-bdc7-632b332e1fd3) -|  (fecc014b-c71b-4bfc-9f55-7997e5084e6f);
\end{tikzpicture}
\end{adjustbox}
\caption{Overview of the proposed quantum computational method to extract properties of free fermionic systems such as the entries of matrices listed in Section \ref{sec:matrixfunctions}. The elements of the construction illustrated here are described in detail in Sections \ref{sec:block_encodings_general}-\ref{sec:observables}.}
\label{fig:overview}
\end{figure*}

\section{Outline} 

In Section \ref{sec:matrixfunctions}, we detail our objects of interest: correlation matrices for the time-evolved and thermal equilibrium states, as well as the Green's function matrix. Each of these objects carries physically meaningful information about the system, and has a form $F(h)$ --- an explicit matrix function of $h$. 

In Sections \ref{sec:block_encodings_general}-\ref{sec:observables} we demonstrate, how the information contained in these matrices can be efficiently extracted from a quantum computer, using the framework of so-called \textit{block-encodings}. Figure \ref{fig:overview} provides a visual scheme, illustrating the structure of our approach. Section \ref{sec:block_encodings_general} explains the block-encoding framework, namely how any $N\times N$ sized matrix $A$ can be encoded into a block of a unitary $U_A$ on $O(n)=O(\log N)$ qubits. We also review the basic tools to produce and manipulate such unitaries $U_A$, which were previously established in the literature. Given the matrix functions $F(h)$ of our interest, we will aim to produce the block-encodings $U_{F(h)}$ as compact quantum circuits. 


The starting point of our circuit construction are smaller unitaries which encode $h$ itself; these unitaries are called sparse access oracles (as the matrix $h$ is required to be sparse). In Section \ref{sec:oraclerealization}, we show how to to implement the sparse access oracles as quantum circuits of size $\mathrm{poly}\log N$. Such implementations are specific to the model of interest: we discuss the cases of $d$-dimensional lattice models and some expander graph geometries; we also demonstrate that quenched disorder can be introduced efficiently. 

In the following Section \ref{sec:blockencodingmatrixfunctions}, having implemented the sparse access oracles for $h$, we move to the second step of the construction --- realizing the block-encoding of matrix functions $F(h)$ of our interest. We detail how this can be done with quantum circuits of size that scales polynomially in $\log N$, as well as polynomially in parameters of the respective function, such as the evolution time $t$, the inverse temperature $\beta$, or the Green's function regularization parameter $\delta^{-1}$. 

Being able to run a circuit which realizes the block-encoding of the matrix $F(h)$ is not sufficient for a successful simulation: one still needs an efficient method to extract physically relevant information from $F(h)$. Section \ref{sec:observables} shows how this can indeed be done, using a Hadamard test and basic sampling techniques. In particular, we demonstrate that the local observables and global densities (such as the total energy density) can be accurately extracted from a block encoding of $F(h)$, while maintaining the $\mathrm{poly}\log N$ complexity of the algorithm.


Sections \ref{sec:complexity} and \ref{sec:speedup} deal with a crucial question: does our approach provide a significant speedup compared to a classical computation? This question can be answered in the affirmative from two perspectives. In Section \ref{sec:complexity}, we take a complexity theory perspective and observe that simulating free-fermionic time dynamics using $\log N$ qubits is BQP-hard. In other words, for a classical computer it is strictly as hard as simulating a general quantum computation on $\log N$ qubits --- which is widely assumed to be exponentially hard in the number of qubits. This establishes that our approach yields an exponential quantum speedup for at least some system geometries. In Section \ref{sec:speedup}, we take a more practical perspective, and focus on the geometries of direct physical interest (such as those given in Section \ref{sec:oraclerealization}). For these models, we compare the performance of our algorithm with the best classical algorithms which are currently available. 
In particular, we find that the quantum algorithm yields a power $(d+1)$ polynomial speedup when simulating the time dynamics of $d$-dimensional lattice models. For simulations of the expander models, we demonstrate an exponential quantum speedup.

We close the main text with the Section \ref{sec:generalizations}, where we sketch how our approach can be generalized to other systems, such as free fermions with pairing terms $(\sim \Delta a_j a_k)$ and free bosons with particle conservation. In Section \ref{sec:discussion}, we discuss the future directions.

\section{Objects of interest} 
\label{sec:matrixfunctions}

We consider three kinds of target objects --- matrix functions of $h$, whose entries encode the physically relevant information.

\begin{itemize}
    \item Correlation matrices $M^{(\beta)}$ of thermal states $\rho_{\beta}=e^{-\beta H}/{\rm Tr}(e^{-\beta H})$ associated with free-fermion Hamiltonians $H$:
    \begin{equation}
    \label{eq:fermi_dirac_main}
        M^{(\beta)}=\frac{I}{I+e^{\beta h}}.
    \end{equation}
    The eigenvalues $n_{\beta}(\epsilon_i)=(1+e^{\beta \epsilon_i})^{-1}$ of $M^{(\beta)}$ correspond to the Fermi-Dirac distribution, with $\epsilon_i$ the eigen-energies of $h$, and $\langle \hat{N}\rangle_{\beta}=\sum_i n_{\beta}(\epsilon_i)$. Note that $h$ here includes a chemical potential term $-\mu \mathbb{I}$, if needed. 
    
    \item Correlation matrices $M(t)$ of time-evolved states $\rho(t)$ (where the time evolution of $\rho(0)$ is under a free-fermion Hamiltonian $H$):
    \begin{equation}
        M(t)=e^{iht}M e^{-iht}, 
    \label{eq:timeevolvedcorrelationmatrix}
    \end{equation}
    with $M$ denoting the correlation matrix of $\rho(0)$. 

    In fact, we will consider a slightly more general object:
    \begin{equation}
        M(t_1,t_2)=e^{iht_1}M e^{-iht_2}, 
    \label{eq:green2}
    \end{equation}
    the entries of which correspond to 
    \begin{equation}
        M_{ij}(t_1,t_2)={\rm Tr}\,(a_i^{\dagger}(t_1) a_j(t_2)\rho),
    \label{eq:greenmain}
    \end{equation}
    with Heisenberg operators $a_i^{\dagger}(t),a_j(t)$ w.r.t. the free-fermion Hamiltonian $H$. 

    Note that for a Hamiltonian $H=H_0+V$ with free-fermionic $H_0$ and interacting perturbation $V$, after applying $U(t)=e^{-i H t}$ to an initial free-fermionic state $\rho$, observables involving creation and annihilation operators can be obtained from $M(t_1,t_2)$ in Eq.~\eqref{eq:green2}. This can be done via a perturbative expansion of $U(t)=e^{-i H t}$ and using Wick's theorem.
    
    \item The Green's function (in the frequency domain) w.r.t. a thermal state $\rho_{\beta}$ of a free-fermion Hamiltonian:
    \begin{multline}
        G^{(\delta,\beta,\omega)}(h) = \frac{\delta}{2}\bigg[\Big( 1-\frac{1}{1+\exp(\beta h)} \Big) \frac{1}{i\delta - (h+\omega)} \\ + \Big( \frac{1}{1+\exp(\beta h)} \Big) \frac{-1}{i\delta + (h+\omega)}\bigg],
    \label{eq:greensomega_main}
    \end{multline}
    with $\delta > 0$ a regularization parameter. 

    $G^{(\delta,\beta,\omega)}(h)$ is a Fourier transform of the time-domain Green's function, the entries of which are given by (here we use time-ordering unlike in Eq.~\eqref{eq:greenmain}):
    \begin{align}
    G_{ij}&(t_1,t_2) = \begin{cases}
      i{\rm Tr}\big( a_{i}^{\dagger}(t_1)a_{j}(t_2)\rho_{\beta} \big), & \text{for }t_1\geq t_2, \\
      -i{\rm Tr}\big( a_{j}(t_2)a_{i}^{\dagger}(t_1)\rho_{\beta} \big), & \text{for }t_1 < t_2,
      \end{cases} \nonumber \\ 
      =&\: \begin{cases}
      \big( ie^{ih(t_1-t_2)}\frac{1}{1+\exp(\beta h)}\big)_{ij}, & \text{for }t_1\geq t_2, \\
      \big(-ie^{ih(t_1-t_2)}\big(1-\frac{1}{1+\exp(\beta h)}\big)\big)_{ij}, & \text{for }t_1 < t_2.
    \end{cases}       
    \end{align}
    The regularization parameter $\delta$ in Eq. \eqref{eq:greensomega_main} ensures that the Fourier transform converges in the case of an isolated system, but can also model interactions with a bath at finite temperature \cite{Altland2010Condensed}. 
\end{itemize}

\section{Block-encodings} 
\label{sec:block_encodings_general}
Let us consider encoding a Hermitian $(N \times N)$-dimensional matrix $A$ into a block of an $n+m$ qubit unitary $U_A$. In general, an $n$-qubit matrix $A$ is said to be block-encoded into $U_A$ if it is equal to the block of $U_A$ where $m$ qubits are in a trivial state, with some coefficient $\alpha$
\begin{align}
    A_{ij}=\alpha \bra{i}_n\bra{0}_m U_A\ket{j}_n\ket{0}_m.
\label{eq:Mij_block_encoding_main}
\end{align}
Here, the matrix indices $i,j\in [N]$ are interpreted as bit-strings of length $n$. The coefficient $\alpha\geq 1$ arises from the fact that $\|U_A\|= 1$ while $A$ is arbitrary. If $\|A\|\leq 1$, we can take $\alpha = 1$. For a useful block-encoding, the coefficient $\alpha$ should not blow up beyond $\mathrm{poly}\log N$. Fortunately, in the applications considered in this work, $\alpha$ will remain a small constant. For the same reasons of maintaining efficiency, we will limit the number of ancillary qubits $m$ to $O(\log N)$. 

We will also allow block-encoding with error $\varepsilon$, the deviation in operator norm between $A$ and $\alpha \bra{0}_m U_A \ket{0}_m$. 
\begin{definition}
    For a matrix $A$ on $n$ qubits and $\alpha, \varepsilon \in \mathbb{R}_+$, an $(m+n)$-qubit unitary $U_{A}$ is an $(\alpha,m,\varepsilon)$-block-encoding of $A$, if
    \begin{align}
        \| A-\alpha (\bra{0}^{\otimes m}\otimes \mathbbm{1}) U(\ket{0}^{\otimes m}\otimes \mathbbm{1})\|\leq \varepsilon.  
    \end{align}
    where $||\cdot||$ is the spectral norm.
    \label{def:block}
\end{definition} 

The quantum circuits that approximately block-encode the matrix functions $F(h)$ are built using elementary circuits $U_h$ that block-encode $h$. These latter block-encodings $U_h$, in turn, contain unitaries which realize so-called \textit{sparse query access} to $h$. To access an $s$-sparse matrix $h$, i.e. a matrix which has up to $s=O(1)$ nonzero entries in any row, we will use `oracle' unitaries $O_r$ and $O_a$ which produce the entries of $h$. The `row' oracle $O_r$ returns, for a given row $i$, all column indices where the matrix $h$ has nonzero entries. The `matrix entry' oracle $O_a$ returns the value of $h$ (given with $n_a$ bits) for a given row and column index. This way, entries of $h$ can be retrieved without explicit access to the $\Theta(2^n)$ nonzero entries of matrix $h$. Let us formally define the \textit{oracle tuple} $\mathcal{O}_{h}$ of a sparse matrix $h$ containing the row and matrix entry oracles, and also their inverses and controlled versions as follows.

\begin{definition}[Sparse Access Oracle Tuple $\mathcal{O}_{h}$]
\label{def:oracles}
    Sparse access for an $s$-sparse $2^{n}\times 2^{n}$ matrix $h$ is defined as 
    \begin{align}
        O_r \ket{i}\ket{0}^{\otimes s(n+1)} = &\: \ket{i}\ket{r(i,1)}\ket{r(i,2)}\ldots\ket{r(i,s)}, \nonumber \\ & \hspace{3.5cm} \forall i\in[2^n],  \nonumber \\
        O_a \ket{i}\ket{j}\ket{0}^{\otimes n_a} = &\: \ket{i}\ket{j}\ket{h_{ij}},~~\forall i,j\in [2^n],
        \label{eq:oracles_main}
    \end{align}
    where $r(i,k)$ is the index for the $k$th nonzero entry of the $i$th row of $h$. Let us now cover a few technicalities.
    $O_{r}$ is a matrix acting on $(s+1)(n+1)$ qubits, and so the first qubit of $\ket{i}$ is in $\ket{0}$.
    To accommodate rows with less than $s$ non-zero entries, one uses the following. If the $i$th row contains $s'<s$ non-zero entries, then the last $(s-s')(n+1)$ qubits are put in the state $\ket{1}\ket{k}$. Note that for states $\ket{r(i,1)}\ldots \ket{r(i,s')}$, the first qubit is in $\ket{0}$.
    $h_{ij}$ is the value of the $(i,j)$th entry of $h$, described by a bitstring with $n_{a}$ binary digits (we will assume this representation to be exact). $O_{a}$ is a matrix acting on $2n+n_{a}$-qubits.

    Furthermore, we define the controlled version of the above sparse access, consisting of
    \begin{align}
        C\text{-}O_r =&\: O_r\otimes \ket{1}\bra{1}_{a} + \mathbbm{1}\otimes \ket{0}\bra{0}_{a}, \nonumber \\
        C\text{-}O_a =&\: O_a\otimes \ket{1}\bra{1}_{a} + \mathbbm{1}\otimes \ket{0}\bra{0}_{a},
    \end{align}
    where each matrix now acts on an additional (ancillary) qubit $a$. We call the collection of six oracles $(O_{r},O_{a},C\text{-}O_{r},C\text{-}O_{a},O_{r}^{-1},O_{a}^{-1},C\text{-}O_{r}^{-1},C\text{-}O_{a}^{-1})$ the sparse access \textit{oracle tuple} $\mathcal{O}_{h}$ of $h$.
\end{definition}

The relation between this definition of the oracle tuple and another common definition is discussed in Appendix \ref{app:def} for completeness. 

Let us now present the following statements, relating the construction of the block-encoding of $h$ and that of polynomials of $h$. Note that these block-encoding constructions contain calls to oracles from the oracle tuple $\mathcal{O}_h$ in Definition \ref{def:oracles}. We shall use these results when constructing the block-encodings of our desired matrix functions. The following statements use Definition \ref{def:block}.

\begin{proposition}[Lemma 48 in \cite{Gilyen_2018_arxiv}]
    A $(s,n+3,\varepsilon_{BE_h})$-block-encoding of $h$, $U_h$ (and its controlled version) consists of $O(1)$ calls to oracles from $\mathcal{O}_{h}$ tuple, $O\big( n+\log^{5/2}(s^2/\varepsilon_{\text{BE}_h}) \big)$ elementary gates and $O\big( sn+n_a+\log^{5/2}(s^2/\varepsilon_{\text{BE}_h}) \big)$ ancillary qubits. Here $n_a$ denotes the number of bits with which the entries of $h$ are specified. 
\label{prop:gilyenUh}
\end{proposition}

\begin{proposition}[Theorem 31 in \cite{Gilyen_2019}]
    Let $p_{d}(x)$ denote a degree-$d$ polynomial s.t. $|p_d(x)|\leq 1/2$ for $x\in [-1,+1]$. Then, a $(1,n+5,4d\sqrt{\varepsilon_{BE_h}/s}+\delta)$-block-encoding of $p_d(h/s)$, $U_{p_d(h/s)}$, consists of $O((n+4)d)$ elementary gates, and at most $d$ calls to unitaries $U_h$, $U_h^{-1}$ or controlled-$U_h$. The classical description of this circuit can be obtained classically in poly$(d,\log(1/\delta))$ time.
\label{prop:gilyenUph}
\end{proposition}

\section{Sparse-access realization for physical systems} 
\label{sec:oraclerealization}
The starting point for our method is to realize the sparse access tuple $\mathcal{O}_h$ for the system Hamiltonian $h$, using efficient quantum circuits. In particular, we need circuit realizations for unitaries $O_r$ and $O_a$ (Eq.~\eqref{eq:oracles_main}); these in fact can be given as (reversible) classical circuits, as no entanglement generation is required. Then the controlled and inverse unitaries from $\mathcal{O}_h$ can also be obtained as efficient circuits (with a constant factor overhead), controlling or inverting the circuits for $O_r$ and $O_a$ gate-by-gate. Please note that `efficient' in our case means $\mathrm{poly} \log N$ gate complexity, i.e., polynomial in the number of qubits rather than the size of $h$. In other words, simply looking up the entries of the $N\times N$ matrix $h$ would not suffice, as that takes time which is exponentially longer than desired. Despite this difficulty, the requirement of the efficient implementation of $\mathcal{O}_h$ can be satisfied for a variety of $h$ of interest.

A large family of free-fermionic models for which the sparse access to $h$ can be efficiently realized are $d$-dimensional tight-binding models. Consider a $d$-dimensional square lattice $\mathcal{L}$ with $L_1\times L_2\times..\times L_d=N_s$ sites, with either periodic or open boundaries. For each site $\vec{x}$, let there be up to $N_0=O(1)$ onsite degrees of freedom such as spin, or local orbital degrees of freedom. We can thus represent each fermionic mode using $n=\big(\Pi_{i=1}^d \lceil \log_2 L_i \rceil \big) \times \lceil \log_2 N_0 \rceil$ qubits as $\ket{\vec{x}=(x_1,\ldots, x_d),o}$ where $N_s=\Theta(2^n)$. 
Inside the lattice, let there be $O(1)$ non-overlapping rectangular domains, modeling different physical regions such as leads versus bulk regions, where parameters in $H$ can be different. We thus consider Hamiltonians of the following form:
\begin{align}
\label{eq:tight_binding_Ham_main}
    H=\sum_{\substack{o_{1},o_2}}\sum_{\substack{\vec{x}\in \mathcal{L},|\vec{t}|_{\rm M}\leq l}} h_{\substack{\vec{x},o_1,\vec{x}+\vec{t},o_2}}~a^\dag_{\vec{x}+\vec{t},o_2} a^{\phantom\dag}_{\vec{x},o_1} +{\rm h.c.},  
\end{align}
    where it is understood (but notationally awkward) that the sum over $\vec{x}\in \mathcal{L},|\vec{t}|_{\rm M}\leq l$ only counts each possible hopping term once. In addition, we have 
    \begin{align}
\label{eq:tb_Ham_dependencies_main}
&h_{\substack{\vec{x},o_1,\vec{x}+\vec{t},o_2}} =g\left(o_1,o_2,D(\vec{x}), D(\vec{x}+\vec{t}), \vec{t}\right),\notag \\
&|h_{\substack{\vec{x},o_1,\vec{x}+\vec{t},o_2}}| \leq 1.
\end{align}
Here $|.|_{\rm M}$ means Manhattan distance in the lattice; the maximal range of the interaction is posited to be constant --- $l=O(1)$. The function $D(\vec{x})$ returns the domain to which $\vec{x}$ belongs: since the domains are rectangular regions, $D(\vec{x})$ can be efficiently computed using standard reversible artihmetic circuits. If $\vec{x}$ or $\vec{x}+\vec{t}$ does not belong to any domain (for example, $\vec{x}+\vec{t}$ is beyond the boundaries of the lattice), the coefficient $h_{\substack{\vec{x},o_1,\vec{x}+\vec{t},o_2}}=0$. Thus, the function $g$ only takes in $O(1)$ information and all $O(1)$ possible nonzero outputs of $g$ can be stored classically, using, say, $O(n_a)$ bits. 
To realize the oracles $\mathcal{O}_h$ from Definition \ref{def:oracles} as poly$(n)$-sized quantum circuits, observe that one can efficiently generate the $O(1)$ input to $g$ and lookup the relevant information.

Going beyond local $d$-dimensional models, we give an example of a model on an \textit{expander graph} which has sparse query access. These graphs have the important property that the number of vertices that lie a distance $d$ away from a given vertex scales exponentially in $d$. Free-fermionic models on such graphs have been a subject of recent interest, especially in the studies of Anderson localization on random regular graphs \cite{tikhonov2016anderson, VA:expander,GarciaMata2017,GarciaMata2020}. In Appendix \ref{app:margulis}, we provide details of the realization of $\mathcal{O}_h$ as poly$(n)$-sized quantum circuits for a simple example: the Margulis expander graph. 

So far, we have proposed models with efficient sparse access where there was only a limited number of possible options for the hopping parameters, and they were input `by hand'. This is in line with a necessary limitation --- even though the system has size $N$, we should be unable to assign every mode an independent value of the hopping parameter. 

However, this restriction can be somewhat relaxed. In particular, one can show that local quenched disorder can also be incorporated into $h$. This has the significance for physics application, as it allows to study Anderson localization. For simplicity, let us focus on realizing onsite disorder in a single domain $D^*$ of a tight-binding model. This means that we introduce a single change to the Hamiltonian of Eqs.~\eqref{eq:tight_binding_Ham_main} and \eqref{eq:tb_Ham_dependencies_main}. Namely, if $D(\vec{x})=D^*$ and $\vec{t}=0$ (both equalities are efficiently checkable), the value of $h_{\vec{x},o_1,\vec{x}+\vec{t},o_2}$ will be replaced by
\begin{align} h_{\vec{x},o_1,\vec{x}+\vec{t},o_2}=\delta_{o_1,o_2}\, \mathrm{PRF}(\vec{x}),
\end{align}
where $\delta_{a,b}$ is the Kronecker symbol and $\mathrm{PRF}$ is a pseudo-random function of the lattice site coordinate $\vec{x}$. Note that a pseudo-random function can be realized as an efficient classical circuit \cite{goldreich1986construct, banerjee2012pseudorandom}. 
Other models of local disorder can be realized similarly. We note that an independent work \cite{chenchan} discusses the application of simulating disordered free fermions in more detail.

\section{Block-encodings of relevant matrix functions} 
\label{sec:blockencodingmatrixfunctions} 

Given the poly$(n)$-effort sparse access tuple $\mathcal{O}_h$, we now aim to realize a block-encoding of the desired matrix functions of $h$ (Section \ref{sec:matrixfunctions}) with an efficient quantum circuit. We will approximate these functions with polynomials of sufficiently low degree, enabling us to use standard methods of block-encodings manipulation (Proposition \ref{prop:gilyenUph}). 

To construct the polynomial approximations, let us first establish the following. Proposition \ref{prop:gilyenUph} prescribes how degree-$d$ polynomials $p_d(x)$ with $x=h/s$ can be block-encoded, with $s$ the sparsity of $h$. We thus require a polynomial approximation $p_d(x)$ to our functions of interest $F(h=sx)$ to be sufficiently accurate in the domain $x\in [-\|h\|/s,+\|h\|/s]$. It can be argued straightforwardly that this domain is at most $[-1,+1]$ by bounding the spectral norm of $h$: 
\begin{proposition}
    Let $h$ denote an $s=O(1)$-sparse Hermitian $N\times N$ matrix with $|h_{ij}| \leq 1, \: \forall i,j$. The spectral norm $||h||/s\leq 1$ by the \href{https://en.wikipedia.org/wiki/Gershgorin_circle_theorem}{Gershgorin circle lemma} (which says that every eigenvalue of $h$ lies within at least one of the $N$ discs  $D_i=\{z\in \mathbb{C}: |z-h_{ii}| \leq \sum_{j\neq i}|h_{ij}| \}$).
\label{lemma:gershgorinmain}
\end{proposition}

To block-encode the thermal correlation matrix in Eq.~\eqref{eq:fermi_dirac_main} and the thermal Green's function in Eq.~\eqref{eq:greensomega_main}, we need to approximate the functions 
\begin{equation}
    f^{(\beta)}(x) := \frac{1}{4}\frac{1}{1+\exp(\beta sx)}
\label{eq:thermalfunction}
\end{equation}
and
\begin{multline}
    g^{(\delta,\beta,\omega)}(x) := \frac{1}{4}\frac{\delta}{2}\bigg[\Big( 1-\frac{1}{1+\exp(\beta sx)} \Big) \frac{1}{i\delta - (sx+\omega)} \\ + \Big( \frac{1}{1+\exp(\beta sx)} \Big) \frac{-1}{i\delta + (sx+\omega)} \bigg]
\label{eq:greensfunctionmain}
\end{multline}
in the domain $x\in [-1,+1]$. These functions have poles in the complex plane at $z = (2k + 1)i\pi/\beta$ (with $k\in \mathbb{Z}$), and at $z = (\pm i\delta-\omega)/s$, respectively. Since these poles might lie in the unit circle for general $\beta$ and $\delta$, we have to resort to polynomial approximation techniques beyond Taylor approximations to obtain a sufficiently accurate approximation for $x\in [-1,+1]$. In particular, we will employ Bernstein's Theorem:

\begin{lemma}[\cite{Bernstein}]
    Let $f(x)$ be analytic on $[-1,+1]$ and analytically continuable to the interior of an ellipse defined by $E_{r} = \{ \frac{1}{2}(z+z^{-1})\colon |z| = r \}$ (for some real-valued $r\geq 1$). Furthermore, let $|f(z)|\leq C$ for $z\in E_{r}$. The error w.r.t. their polynomial approximation $p_{d}(x)$ (Chebyshev expansion truncated at degree $d$) can be bounded as
    \begin{equation}
        \max_{x\in [-1,+1]}|f(x)-p_{d}(x)| \leq \frac{2Cr^{-d}}{r-1}.
    \end{equation} 
\label{theorem:Bernstein}
\end{lemma}

Using this result, we derive the following error bounds for the polynomial approximations of Eqs. \eqref{eq:thermalfunction} and \eqref{eq:greensfunctionmain}. The proofs of Lemmas \ref{lem:fd_poly_apx_main} and \ref{lem:greens_poly_apx_main} are given in Appendices \ref{app:fm} and \ref{sec:greensfunction}. 

\begin{lemma}[Simplified version of Lemma \ref{lem:fd_poly_apx} in Appendix \ref{app:fm}]
\label{lem:fd_poly_apx_main}
For the function $f^{(\beta)}(x)$ in Eq.~(\ref{eq:thermalfunction}) (with $\beta,s\geq 0$), one can efficiently construct a polynomial $p_{d}(x)$ of degree $d$ such that
\vspace{-0.7cm}
\begin{center}
\begin{align}
    & {\rm max}_{x\in [-1,+1]} |f^{(\beta)}(x)-p_{d}(x)| \leq {\rm poly}(\beta s)/d. 
\end{align}
\end{center}
\end{lemma}

\begin{lemma}[Simplified version of Lemma \ref{lem:greens_poly_apx} in Appendix \ref{sec:greensfunction}]
\label{lem:greens_poly_apx_main}
For the function $g^{(\delta,\beta,\omega)}(x)$ in Eq.~ \eqref{eq:greensfunctionmain} (with $\beta,\delta,s>0$), one can efficiently construct a polynomial $p_{d}(x)$ of (even) degree $d$ such that
\vspace{-0.7cm}
\begin{center}
\begin{align}
    &{\rm max}_{x\in [-1,+1]} |g^{(\delta,\beta,\omega)}(x)-p_{d}(x)| \nonumber \\ &\:\leq \big({\rm poly}(\beta s) + {\rm poly}(s/\delta)\big)/d. 
\label{eq:greenspolyapprox}
\end{align}
\end{center}
\end{lemma}

Combining Lemmas \ref{lem:fd_poly_apx_main} and \ref{lem:greens_poly_apx_main} with Propositions \ref{prop:gilyenUh} and \ref{prop:gilyenUph}, we directly obtain Theorems \ref{lem:fermi_dirac_main} and \ref{lem:greensfunction_main} below. The detailed proofs are given in Appendices \ref{app:fm} and \ref{sec:greensfunction}. Note that -- crucially, because of the factors $\frac{1}{4}$ in Eqs. \eqref{eq:thermalfunction} and \eqref{eq:greensfunctionmain} -- the polynomials $p_d(x)$ that are block-encoded obey $|p_d(x)|\leq 1/2$ for $x\in [-1,+1]$, provided that the error of the polynomial approximation is $O(1)$. The size of the circuits that block-encode $M^{(\beta)}(h)$ in Eq. \eqref{eq:fermi_dirac_main} and $G^{(\delta,\beta,\omega)}(h)$ in Eq. \eqref{eq:greensomega_main} is poly$(n)$, provided that $\beta,1/\delta = {\rm poly}(n)$ and when the oracles from $\mathcal{O}_{h}$ are poly$(n)$-sized circuits (such as those in Section \ref{sec:oraclerealization}).

\begin{theorem}[Block-encoding of the thermal correlation matrix Eq.~\eqref{eq:fermi_dirac_main}]
\label{lem:fermi_dirac_main}
For an $s$-sparse Hamiltonian $h$ on $n$ qubits, assume access to the oracle tuple $\mathcal{O}_{h}$. We denote the controlled $(1,n+5,\varepsilon)$-block-encoding of $\frac{1}{4}M^{(\beta)} = \frac{1}{4}I/(I+\exp(\beta h))$ by $C$-$U_{M^{(\beta)}}$. 
The implementation of this block-encoding for $\beta = {\rm poly}(n)$ requires ${\rm poly}(n)/\varepsilon$ calls to oracles from the oracle tuple $\mathcal{O}_h$, $O(n)+n_a+\log^{5/2}({\rm poly}(n)/\varepsilon^4)$ ancillary qubits and $O(n)+{\rm poly}(n)/\varepsilon+\log^{5/2}({\rm poly}(n)/\varepsilon^4)$ additional elementary gates. To implement this block-encoding, an additional classical computing time of $\text{\normalfont{poly}}\big(n/\varepsilon,\log(1/\varepsilon)\big)$ is required. 
\end{theorem} 

\begin{theorem}[Block-encoding of the thermal Green's function Eq.~\eqref{eq:greensomega_main}]
\label{lem:greensfunction_main}
For an $s$-sparse Hamiltonian $h$ on $n$ qubits, assume access to the oracle tuple $\mathcal{O}_{h}$. We denote the controlled $(1,n+5,\varepsilon)$-block-encoding of $\frac{1}{4}\:G^{(\delta,\beta,\omega)}(h)$ in Eq.~\eqref{eq:greensomega_main} by $C$-$U_{G^{(\delta,\beta,\omega)}}$. 
The implementation of this block-encoding for $\beta,1/\delta = {\rm poly}(n)$ requires poly$(n)/\varepsilon$ calls to oracles from the oracle tuple $\mathcal{O}_h$, $O(n)+n_a + \log^{5/2}\big( {\rm poly}(n)/\varepsilon^4 \big)$ ancillary qubits and $O(n) + {\rm poly}(n)/\varepsilon + \log^{5/2}\big( {\rm poly}(n)/\varepsilon^4 \big)$ additional elementary gates. To implement this block-encoding, an additional classical computing time of $\text{\normalfont{poly}}\big(n/\varepsilon,\log(1/\varepsilon)\big)$ is required. 
\end{theorem}

Next, let us focus on block-encoding the time-evolved correlation matrix $M(t)$ in Eq. \eqref{eq:timeevolvedcorrelationmatrix}. To block-encode it, we will use a block-encoding of $\exp(iht)$ as a sub-routine. The construction of this latter block-encoding through polynomial approximations is already considered in \cite{Low_2019,Gilyen_2018_arxiv}, and we will use this construction from \cite{Gilyen_2018_arxiv} directly. We construct a block-encoding of $M(t)$ using the product of block-encodings of $\exp(iht)$, an initial correlation matrix $M$ and $\exp(-iht)$. A detailed proof of Theorem \ref{lem:time_evo_main} is given in Appendix \ref{sec:time-evolv}. There, we in fact consider a block-encoding of the more general object $M(t_1,t_2)$ in Eq. \eqref{eq:green2}. 

\begin{theorem}[Simplified version of Theorem \ref{lem:time_evo} in Appendix \ref{sec:time-evolv}:  Block-encoding of the time-evolved correlation matrix in Eq. \eqref{eq:timeevolvedcorrelationmatrix}]
\label{lem:time_evo_main}
For an $s$-sparse Hamiltonian $h$ on $N$ fermionic modes, assume access to the oracle tuple $\mathcal{O}_h$. In addition, assume access to the $(\alpha,m,\varepsilon_M)$-block-encoding $U_M$ of a correlation matrix $M$ of a fermionic state on $N$ modes.  
The $\big(\alpha, 2n+m+10,\varepsilon+\varepsilon_M\big)$-block-encoding $U_{M(t)}$ of $M(t)=e^{iht}Me^{-iht}$
can be produced using $D(\alpha,\varepsilon, t) = O\big(|t| + \log(\alpha/\varepsilon) \big)$ calls to oracles from the tuple $\mathcal{O}_{h}$, and a single use of the block-encoding $U_{M}$. Moreover, one uses $O\big( n|t| + \log(\alpha/\varepsilon) 
+ D(\alpha,\varepsilon, t)\big(n+\log^{5/2}(\alpha|t|/\varepsilon)\big)\big)$ elementary gates and $O\big(n_{a} + \log^{5/2}(\alpha|t|/\varepsilon)\big)$ ancillary qubits.
\end{theorem}

Combining Lemma \ref{lem:obs_extraction} with Theorems \ref{lem:fermi_dirac_main}, \ref{lem:greensfunction_main} and \ref{lem:time_evo_main}, we can respectively estimate entries of $M^{(\beta)}$ in Eq.~ \eqref{eq:fermi_dirac_main}, $G^{\delta,\beta,\omega}$ in Eq. \eqref{eq:greensomega_main}, and $M(t)$ in Eq. \eqref{eq:timeevolvedcorrelationmatrix}, up to $1/\mathrm{poly}\,(n)$ error with poly$(n)$ effort. Note that --- asymptotically --- the circuit implementing the controlled block-encodings (which is required for the Hadamard test in Lemma \ref{lem:obs_extraction}) is of the same size as the block-encoding circuits themselves.

\section{Extracting observables}
\label{sec:observables}
Having explicitly constructed $(\alpha,m,\varepsilon)$-block-encodings $U_{F(h)}$ of our objects of interest $F(h)$, let us detail how to extract relevant observables from such block-encoding unitaries. If $U_{F(h)}$ is given as a $\mathrm{poly}(n)$-sized quantum circuit, the real and imaginary parts of $F(h)_{ij}$ can be extracted efficiently using the so-called \href{https://en.wikipedia.org/wiki/Hadamard_test}{Hadamard test} using an ancillary-qubit-controlled-$U_{F(h)}$. Note that the circuit size required to implement controlled-$U_{F(h)}$ scales the same as $U_{F(h)}$, up to a constant factor overhead. We can extract $F(h)_{ij}$ with an accuracy specified in the next Lemma \ref{lem:obs_extraction}. This Lemma is stated for a general block-encoding unitary and is proved in Appendix \ref{sec:cor-estim}. From Lemma \ref{lem:obs_extraction} it is clear that the error up to which $F(h)_{ij}$ can be estimated is $1/{\rm poly}(n)$, since we allow for at most poly$(n)$ calls to the block-encoding unitaries. 

\begin{lemma}
\label{lem:obs_extraction}
Given an $n$-qubit matrix $A$. Let $C\text{-}U_A$ (acting on $n + m + 1$ qubits) denote the controlled version of the $(\alpha,m,\varepsilon)$-block-encoding $U_A$ of $A$. 
An estimate $\hat{A}_{ij}$ of entry 
$A_{ij}$ can be obtained s.t. $\bigl\lvert \hat{A}_{ij} - A_{ij} \bigr\rvert \leq \varepsilon+\alpha\tilde{\varepsilon}$ with probability at least $1-\delta$, using ${\rm poly}(n)$-sized circuits and at most $D(\tilde{\varepsilon}, \delta) = \Theta\big(\tilde{\varepsilon}^{-2}\log(4\delta^{-1})\big)$ calls to $C\text{-}U_A$.
\end{lemma}

We note that in case when $F(h)=M$ is a correlation matrix and $H$ corresponds to a lattice model, one can also obtain correlation matrix entries in momentum space --- by using $U_M$ and the efficient Quantum Fourier Transform circuit \cite{book:NC}. 

Going beyond individual matrix elements, for any local fermionic Hamiltonian term $H_x$ in $H$, for example $H_x=\left(h_{ij} a^\dag _j a_i+h^*_{ij}a^\dag _i a_j\right)$ (with $|h_{ij}|\leq 1$) or $H_x=\left(V_{ijkl} a^\dag _i a^\dag_j a_k a_l + V^*_{ijkl} a^\dag _l a^\dag_k a_j a_i\right)$ (with $|V_{ijkl}|\leq 1$), the expectation of that term w.r.t. a state $\rho$ can be efficiently extracted from the block-encoding of its correlation matrix $U_M$ \footnote{From this point onwards, all considered states are free-fermionic, unless stated otherwise.}. In this way one can also obtain the total energy density of $\rho$ relative to a system Hamiltonian $H$. To do so, one can sample from the Hamiltonian terms uniformly at random and evaluate the expectation value of individual terms as mentioned above. For $H$ being a free-fermion Hamiltonian, this sampling can be implemented using the sparse access model discussed below; this method of sampling can be extended to interacting Hamiltonians. 
We can obtain the following concentration bound on this evaluated energy density $e$, assuming, for simplicity, that the expectation of an individual term is learned from $U_{M}$ without error. By assumption, we have that $\bigl\lvert \mathrm{Tr}(H_x \rho) \bigr\rvert \leq 1$ for each Hamiltonian term $H_x$. This allows us to infer the Chernoff bound, which says that for sample size $S = \Theta\big(\varepsilon^{-2}\log(\delta^{-1})\big)$, we have
\begin{align}
        \mathbb{P}\Big(\bigl\lvert e-\mathrm{Tr}\big(H\rho \big)\bigr/K\rvert \leq \varepsilon \Big)\geq 1-\delta,
\end{align}
where $K = \Theta(2^n)$ is the number of terms in the Hamiltonian $H$. Similarly, densities of other Hermitian operators can be learned through sampling, such as the particle density $\langle \hat{N} \rangle/2^{n} = \text{Tr}(M)/2^{n}$.

\section{Complexity} 
\label{sec:complexity}
We have presented a method for simulating free-fermionic systems on $N=2^n$ modes with polynomial resources in $n$ in a variety of settings. The naive classical treatment of $2^n$ fermionic modes, on the other hand, requires exponential resources. Therefore, the naive speedup of our quantum method is exponential.
However, our approach comes with manifest qualifications, namely the requirement for the oracle tuple $\mathcal{O}_h$ to be implementable using poly$(n)$-sized quantum circuits, time dynamics being simulable only for time $t=\mathrm{poly}\,(n)$, thermal states for $\beta=\mathrm{poly}\,(n)$ and Green's functions for $\beta,1/\delta=\mathrm{poly}\,(n)$. Competing classical approaches could hypothetically exploit this structure of our setting. To settle this issue, one can readily argue that our method generally yields an exponential quantum speedup, by showing that it solves a BQP-complete problem. Roughly speaking, BQP-complete problems are the hardest problems which can be efficiently solved by a quantum computer \cite{BV}. 
Since for single-particle dynamics, the character of the particle, --be it a boson, fermion or distinguishable particle-- is not relevant,
BQP-hardness of time-dynamics follows in principle from Theorem 3 in \cite{babbush2023exponential}, using techniques such as those developed in Ref.~\cite{nagaj:phd}. For completeness, we provide a slightly different proof for the complexity of the evolution of a multi-particle fermionic state in Appendix \ref{sec:bqp}.

\begin{theorem}
    \label{thm:time_evo_main}
    Let $\rho_0$ be a (multi-particle) fermionic state on $2^n$ modes, such that its correlation matrix $M_0$ is sparse, and the access oracle tuple $\mathcal{O}_{M_0}$ can be implemented as a ${\rm poly}(n)$-sized quantum circuit.  Given a quadratic Hamiltonian $H$ on $2^n$ modes, let $h$ be as in Eq.~\eqref{eq:defhmain} and sparse, and we assume that the oracle tuple $\mathcal{O}_{h}$ is implemented as a ${\rm poly}(n)$-sized quantum circuit. For $t={\rm poly}(n)$, the problem is to decide whether, for some given mode $j$, $n_j(t)={\rm Tr}\,(a_j^{\dagger}a_j e^{-iHt}\rho_0 e^{i H t})\geq 1/p(\sqrt{n})$ (with $p$ a polynomial) or $\leq \exp(-\sqrt{n})$,  given a promise that either one is the case. This problem is BQP-complete.
\end{theorem}

\section{Quantum speed-up in a variety of settings} 
\label{sec:speedup}
We have established that our algorithms in principle provide an exponential speed-up, at least in the setting of time evolution. In this section, we argue what the speed-up is for several models of physical importance. To that end, let us first argue that for $d = O(1)$-dimensional lattice models, entries of our matrix functions of interest (see Section \ref{sec:matrixfunctions}) can be estimated {\em classically} with poly$(n)$ effort for $\beta,1/\delta,t = {\rm poly}(n)$. 

Lieb-Robinson bounds \cite{hastings, Chen_2023, tran-LR} imply that the time evolution of observables such as the occupation number of a mode $i$ at some position (starting from a product state with some modes occupied and others unoccupied) is only affected by $O(t^d) = {\rm poly}(n)$ sites in a ball of radius proportional to $t$ around that position. Similarly, Ref.~\cite{hastings} shows that, for a given mode $i$, the thermal correlation matrix entries $|M^{(\beta)}_{ij}|$ decay exponentially with the distance between modes $i$ and $j$, with a characteristic length $O(\beta)$. Mode $i$ is therefore only non-trivially correlated with $O(\beta^{d}) = {\rm poly}(n)$ modes in a ball of radius $O(\beta)$ around it. This latter fact suggests that an entry $M^{(\beta)}_{ij}$ can be classically evaluated with ${\rm poly}(n)$ effort, provided that $\beta={\rm poly}(n)$. Let us formalize this as follows. 

\begin{table*}[t]
\centering
{\renewcommand{\arraystretch}{1.3}
 \begin{tabular}{|| c | c | c | c ||} 
 \hline 
  & \begin{tabular}{@{}c@{}}$d = O(1)$-dim. \\ lattice models \end{tabular} & \begin{tabular}{@{}c@{}}expander \\ graphs \end{tabular} & \begin{tabular}{@{}c@{}}general sparse \\ models \end{tabular} \\ [1.5ex]
 \hline\hline
 Classical algorithms $^*$ & \multicolumn{3}{c||}{$r_{\text{prep},C}+N\cdot t$} \\ [0.5ex] 
 \hline
 Quantum algorithms & \multicolumn{3}{c||}{$r_{\text{prep},Q}+\text{poly}\log(N)\cdot t$} \\ [0.5ex]
 \hline
 Lieb-Robinson time & $N^{1/d}$ & $\log(N)$ & - \\ [0.5ex]
 \hline
 Speedup & \hspace{0.2cm} power-($d+1$) polynomial $^{**}$ \hspace{0.2cm} & \hspace{0.2cm} exponential $^{**}$ \hspace{0.2cm} & \hspace{0.2cm} exponential $^{***}$ \hspace{0.2cm} \\ [0.5ex]
 \hline
 \end{tabular}
 }
 \caption{Asymptotic run-times for evaluating entries of time-evolved correlation matrices (with $1/\text{poly}(n)$ error) for three different system types: lattice models, expander graphs and general sparse models. For the former two, we start from a thermal correlation matrix at $\beta = O(1)$ (of some $h'$, different from $h$ used for time evolution). For the latter, we start from a Slater determinant (free fermion pure) state. The third row gives the Lieb-Robinson time (only denoted for lattice models and expander graphs), which corresponds to the time it takes the Lieb-Robinson light cone to contain the entire system. The run-times of classical algorithms (for evolutions over a time interval which is at least the Lieb-Robinson time) and our quantum algorithms are given. In addition, we provide the associated speedups for the lattice models and expander graphs at the Lieb-Robinson time, and the speedup for general sparse models at $t = \text{poly}(n)$. The run-times required to prepare the starting state are denoted by $r_{\text{prep},C}$ and $r_{\text{prep},Q}$, for respectively the classical and quantum algorithms. Note that $r_{\text{prep},Q} = \text{poly}(n)$ in all three scenarios. $^*$Run-times of ---to the best of our knowledge--- the best classical algorithm for these applications \cite{Costa_2019}. $^{**}$Speedups compared to the aforementioned classical algorithms. $^{***}$Speedup assuming that it takes exponential (in $n$) time to solve BQP-complete problems classically.}
 \label{table:runtimes_main}
\end{table*}

\begin{lemma}
\label{lem:classicalsim}
    Let $h\in \mathbb{C}^{2^n \times 2^n}$ be an $s = O(1)$-sparse matrix that corresponds to a $d=O(1)$-dimensional lattice model, cf. Eq.~\eqref{eq:defhmain} with entries as in Eqs.~(\ref{eq:tight_binding_Ham_main}) and (\ref{eq:tb_Ham_dependencies_main}). Assume ${\rm poly}(n)$-effort classical access to the oracles $O_r$ and $O_a$ (see Definition \ref{def:oracles}) for $h$. Let $F(h)$ be a matrix function of $h$. If $\max_{x\in [-1,+1]}|F(x) - p_K(x)| \leq {\rm poly}(n)/K$ with $p_K(x)$ a degree-$K$ polynomial, then an entry $F(h)_{ij}$ can be estimated with that same error using ${\rm poly}(K)\times {\rm poly}(n)$ classical effort. For some $K = {\rm poly}(n)$, the error thus becomes $1/{\rm poly}(n)$ with ${\rm poly}(n)$ classical effort. 
\end{lemma}
\begin{proof}
    If one is able to estimate $\bra{i}h^{k}\ket{j}$ for any $k\in\{0,1,\ldots,K\}$ with effort $E$, then $\bra{i}p_K(h)\ket{j} = \sum_{k=0}^{K}\alpha_k \bra{i}h^{k}\ket{j}$ can be evaluated with effort $K\times E$. By assumption, $\bra{i}F(h)\ket{j}$ can then be classically approximated up to ${\rm poly}(n)/K$ error with $K\times E$ effort. Since $h$ corresponds to a $d=O(1)$-dimensional lattice model, $h^{k}\ket{j}$ is only supported on $O(k^d) = {\rm poly}(k)$ $\ket{i}$'s. We can thus evaluate each $\bra{i}h^{k}\ket{j}$ for $k\in\{0,1,\ldots,K\}$ using ${\rm poly}(k)$ calls to the oracles and with a total $E = {\rm poly}(k)\times {\rm poly}(n)$ computational effort. Therefore, $\bra{i}F(h)\ket{j}$ can be approximated classically with ${\rm poly}(n)/K$ error with $K\times {\rm poly}(k)\times {\rm poly}(n) = {\rm poly}(K)\times {\rm poly}(n)$ effort. Clearly, there is a $K = {\rm poly}(n)$ so that the error becomes $1/{\rm poly}(n)$ and which yields a ${\rm poly}(n)$ classical effort. 
\end{proof}

Combined with Lemmas \ref{lem:fd_poly_apx_main} and \ref{lem:greens_poly_apx_main}, Lemma \ref{lem:classicalsim} implies the following for $d=O(1)$-dimensional lattice models. In the parameter regimes of Theorems \ref{lem:fermi_dirac_main} and \ref{lem:greensfunction_main}, entries of the thermal correlation matrix in Eq.~\eqref{eq:fermi_dirac_main} and of the thermal Green's function in Eq.~ \eqref{eq:greensomega_main} can be estimated  up to $1/{\rm poly}(n)$ error with poly$(n)$ classical effort. 

Using similar reasoning, entries of the time-evolved correlation matrix $M(t)$ in Eq. \eqref{eq:timeevolvedcorrelationmatrix} can be evaluated classically with poly$(n)$ effort for $t = {\rm poly}(n)$. In fact, assuming exact classical access to entries $\bra{k}M\ket{l}$ of an initial correlation matrix $M$ for given $(k,l)$, one can obtain entries $M(t)_{ij}$ with $1/\exp(n)$ error. The improved error scaling comes from the fact that the polynomial approximation error of $\exp(ith)$ can be bounded by $1/\exp(n)$ even for degree $K = {\rm poly}(n)$, provided that $t = {\rm poly}(n)$. A detailed treatment is given in Appendix \ref{sec:classicalsim}. Note that if we apply the time evolution to $M'^{(\beta)}$ (where $M'^{(\beta)}$ is the thermal correlation matrix corresponding to some $h'\neq h$), the accuracy reduces to $1/{\rm poly}(n)$ due to the error in estimating entries of $M'^{(\beta)}$. 

Despite losing the exponential speed-up for $d=O(1)$-dimensional lattice models, let us argue that we retain a power-$(d+1)$ polynomial speed-up for such models. Let us focus on the task of estimating entries of the time-evolved correlation matrix from Eq.~ \eqref{eq:timeevolvedcorrelationmatrix}. In particular, let us focus on the task of time-evolution for $t$ proportional to the Lieb-Robinson time $t_{LR}$, which is the time it takes for a Lieb-Robinson light cone to contain the entire system. For lattice models, $t_{LR} = N^{1/d}$. To then compute an entry of the correlation matrix $M(t)=e^{iht} M e^{-iht}$, known classical algorithms require $\Omega(Nt)=\Omega(t^{d+1})$ run-time \cite{Costa_2019}. Given the ${\rm poly} \log (N) \cdot t$ runtime of our quantum algorithm, we obtain a power-$d+1$ polynomial speedup. In particular, this yields a cubic speedup for $d=2$ lattices and quartic speedup for $d=3$
---which can be of interest in early fault-tolerant devices \cite{babbush-beyond-quadratic}. 

Crucially, our method can also be applied to settings other than lattice models, and the exponential speedup for those settings can be maintained. Let us consider tight-binding models on expander graphs, such as the Margulis graph considered in Section \ref{sec:oraclerealization}. The Lieb-Robinson time, due to the expansion property of the graph, will be logarithmic in the number of modes $N$: $t_{LR} = \log(N)$. We note that light cones also grow rapidly in other graphs with \textit{log-sized diameter}, such as the hyperbolic lattices (see \cite{kollar} for recent studies of such tight-binding models). 
We expect to recover the full exponential quantum speedup for their simulation because at $t_{LR} = \log(N)$, the quantum run-time is poly$\,\log(N)$ while known classical algorithms have run-time $\Omega(Nt)=\Omega(N\log(N))$ \cite{Costa_2019}. This speed-up can be of particular interest for, e.g., the study of Anderson localization on expander graphs \cite{tikhonov2016anderson, VA:expander, GarciaMata2017,GarciaMata2020}.

To summarize the quantum advantage in different problem settings, Table \ref{table:runtimes_main} gives an overview of the asymptotic run-times of classical algorithms and our quantum algorithms, and associated quantum speedups for the problem of time-evolution.

\section{Generalizations} 
\label{sec:generalizations}
The time-evolution framework presented in this paper can be made more general and applied to systems beyond free fermions. In a general quantum system described by a Hamiltonian $H$, one can consider a $N$-sized set of operators $\{O_j\}$ such that $[H,\, O_k]=\sum^{N}_{j=1}h_{jk}O_j$. This is sufficient for a matrix $M_{jk}=\mathrm{Tr}\big(\rho O^\dag_j O_k\big)$ to transform as $M\mapsto e^{-iht} M e^{iht}$ under time evolution. Further assuming that $h$ is a hermitian matrix, this allows treatment of $M$ as a block-encoding of the type considered in this work. Beyond the free-fermionic systems on which we focused in this work, this general framework admits fermionic $H$ which include pairing $(\Delta a_j a_k + h.c.)$ terms. In this case the relevant set $\{O_j\}$ would include not just annihilation but also creation operators. Another example is a system of $2^n$ free bosons with particle conservation, in which case $\{O_j\}$ should be chosen as bosonic annihilation operators. Beyond $\mathrm{Tr}\big(\rho O^\dag_j O_k\big)$, one can consider $M_{j_1,..,j_l;k_1,..,k_{l'}}=\mathrm{Tr}\big(\rho O^\dag_{j_1}..O^\dag_{j_l} O_{k_1}..O_{k_{l'}}\big)$, which can be considered as a rectangular matrix acting on $n\cdot \mathrm{max}(l,l')$ qubits, and block-encoded accordingly. The time evolution of these objects is defined similarly to that of $M_{jk}$, and therefore can be easily found as a block-encoding, given the block-encoding of the initial state. The flexibility of this general block-encoding framework is comparable to the one based on `shadow' states, presented in Ref.\,\cite{GoogleShadow} (see Appendix \ref{sec:alt_encodings} for a discussion of the differences). 

\section{Discussion} 
\label{sec:discussion}
In this work, we develop quantum algorithms that solve several free fermion problems. We discuss in detail what type of speedup is achieved over classical algorithms and present generalizations of our approach. 

One obvious avenue for future research is to apply our method to other matrix functions of $h$. For example, one should be able to estimate the free energy density of a $2^n$-mode free-fermion system $\frac{F}{2^n}=-(\beta 2^n)^{-1}\log {\rm Tr}\,(e^{-\beta H})=-(\beta 2^n)^{-1}{\rm Tr}\,(\log(I+e^{-\beta h}))$ with error $\varepsilon$, using a polynomial approximation of the function $\log(I+e^{-\beta h})$ for $\beta={\rm poly}(n)$, the block-encoding of $h$, and sampling entries to model the trace function. Using an estimate of the free energy density $F/2^n=(\langle H \rangle_{\beta}-\beta^{-1}S(\rho_{\beta}))/2^n$, one can in turn estimate an entropy density, given an energy density estimate, or a derivative of $F/2^n$ with respect to $\beta$ such as the specific heat. Another possible generalization of our work is a ${\rm poly}(n)$-efficient estimation of matrix elements or observable expectations due to free-fermionic {\em dissipative} dynamics, which was shown to be classically simulatable in $O(2^{3n})$ time in \cite{BK:dissip}. 

One could also consider how block-encoding techniques fare when applied to estimating entries of a free-bosonic thermal correlation matrix $M^{(\beta)}=I/(e^{\beta h}-I)$ of Bose-Einstein form. A block-encoding of a polynomial approximation as developed in Lemma \ref{lem:fd_poly_apx_main} and Theorem \ref{lem:fermi_dirac_main} in Section \ref{sec:blockencodingmatrixfunctions} requires a ${\rm poly}(n)$ bound on the mode occupation number (so that the matrix function be block-encoded), which can however grow as large as the number of particles for a Bose-Einstein condensate. Mathematically, the Bose-Einstein distribution with $\epsilon_i\geq 0$ has a singularity at $\epsilon_i =0$ which has to be avoided (by choosing a small enough chemical potential $\mu$) in order to place any bound.
Note that similar points about only algebraic speed-ups for local lattice models (Lemma \ref{lem:classicalsim}) were made for bosonic/oscillator systems in a more recent work \cite{sakamoto2025}.

Another outstanding open direction is to compute and optimize the precise implementation overhead and circuit depth for our proposed algorithms, as applied to simulation problems of practical interest. 

Let us point out an open question in the setting of time-dynamics on $2^n$ fermionic modes (cf. Eq.~\eqref{eq:timeevolvedcorrelationmatrix}). One task that can be performed with poly$(2^n)$ classical effort \cite{TD:freefermion} is computing the overlap 
\begin{multline}
    |\bra{S_1}\exp(-itH)\ket{S_2}|^2 =\\ {\rm Tr}\big[ \underbrace{\exp(-itH)\ket{S_2}\bra{S_2}\exp(itH)}_{\ket{S_3}\bra{S_3}} \ket{S_1}\bra{S_1}\big],
\label{eq:overlap}
\end{multline}
with $\ket{S_1}$ and $\ket{S_2}$ single-Slater determinant states and $H$ a free fermion Hamiltonian as in Eq.~\eqref{eq:defhmain}, and therefore $\ket{S_3}$ is also a Slater determinant state. If $\ket{S_1}\bra{S_1}$ (for simplicity) is a standard-mode-basis Slater determinant state, then it can be expressed as a product of $2^{n+1}$ creation and annihilation operators. Using Wick's theorem, evaluating this weight-$2^{n+1}$ correlator in Eq.~\eqref{eq:overlap} requires evaluating products of $2^n$ entries of the correlation matrix (cf. Eq.~\eqref{eq:defMmain}) of state $\ket{S_3}$. This task -- at least with naive attempts -- cannot be performed using our methods with poly$(n)$ quantum effort, since we can only evaluate poly$(n)$ entries of the time-evolved correlation matrix, although approximate sampling methods could come into play.

\section{Acknowledgements}
We thank C. Beenakker, A. Bishnoi, J. Helsen, T.E. O'Brien, M. Pacholski, S. Polla, K.S. Rai, R. Somma, A. Ciani and A. Montanaro for insightful discussions and feedback. This work is supported by QuTech NWO funding 2020-2026 – Part I “Fundamental Research”, project number 601.QT.001-1, financed by the Dutch Research Council (NWO).  Y.H. acknowledges support
from the Quantum Software Consortium (NWO Zwaartekracht).


%


\appendix

\section{Alternative Encodings} 
\label{sec:alt_encodings}

In this section we describe alternative ways of representing a fermionic correlation matrix using qubits and their potential drawbacks.

A compressed representation of free-fermionic states on $2^n$ modes, as well as their dynamics, is readily obtained by using a (mixed) quantum state $\sigma=M/{\rm Tr}(M)$ of $n$ qubits to represent the normalized correlation matrix of $\rho$. One then computes, ---evolves and measures---, with $\sigma$ to learn properties of $\rho$ or its time-dynamics. For pure single-particle free-fermionic states $\rho$, $\sigma$ is a rank-1 projector, and $\sigma$ projects onto the bitstring $\ket{i}$ when $\rho$ corresponds to $a_i^{\dagger}\ket{\rm vac}$, $i=0,\ldots, N-1$ where $\ket{\rm vac}$ is the fermionic vacuum state. Once a state $\sigma$ is prepared, its time-evolution can readily be simulated: when $\rho$ evolves via $e^{-i H t}$ with free-fermion Hamiltonian $H$, $\sigma \rightarrow e^{i ht} \sigma e^{-i ht}$. Sparse oracle access to $h$ ---see Definition \ref{def:oracles}--- then allows for the efficient implementation of time-evolution in terms of its dependence on $t$ and calls to the oracle \cite{Low_2019, lin2022lecturenotesquantumalgorithms}, starting from some easy-to-prepare initial state. For example, the initial state could be a set of fermions in a subset $S$ of $2^m$ modes $\ket{i}$ (such that an efficient classical circuit can map $S$ onto the set of $m$-bitstrings), or a subset of modes in the Fourier-transformed basis (as the QFT is an efficient quantum circuit). One can also adapt the heuristic quantum Metropolis-Hastings algorithm \cite{Temme_2011,QMS:irani} to the Fermi-Dirac distribution and sparse Hamiltonians $h$, since the algorithm uses quantum phase estimation for $e^{i ht}$ at its core.
Even though the algorithm converges to the thermal state $\sigma_{\beta}=M^{(\beta)}/{\rm Tr}(M^{(\beta)})$, ${\rm poly}(n)$-efficiency is not guaranteed and unlikely for low-enough temperature.

Given a state $\sigma$, one can apply any learning algorithm for $n$-qubit states. For example, one can use shadow tomography \cite{HKP:shadow} to estimate the expectation of $L$ observables, such as $O_k=\ket{k}\bra{k}, O_{lk}^R=\ket{l}\bra{k}+\ket{k}\bra{l},O_{lk}^L=i(\ket{l}\bra{k}-\ket{k}\bra{l})$, with computational effort $O(\log (L))$ using random Clifford circuits of ${\rm poly}(n)$ size. 

There are a few disadvantages to this simple and direct method of representing the state via its correlation matrix. It is not immediately obvious how to estimate a time-dependent correlation function as in Eq.~\eqref{eq:greenmain} as it relates to measurements on $e^{i h t_1} \sigma e^{-i h t_2}$ which is not a state. Second, and more crucially, any learning of a linear function of $\sigma$ with accuracy $\varepsilon$, leads to learning with accuracy $\varepsilon \,{\rm Tr}(M)=\varepsilon \langle \hat{N} \rangle$ for the correlation matrix $M$ itself. Therefore one expects poor accuracy for large particle number $\langle \hat{N} \rangle$; this in particular makes it impractical to extract individual matrix elements.

Thus in the main text of this paper we choose not to directly encode a correlation matrix as a quantum state, but rather apply quantum computational block-encoding techniques.

Recently, Ref.~\cite{GoogleShadow} introduced a general quantum simulation framework with compressed `shadow' quantum states with applications to free bosons and free fermion systems. We note that the results in Ref.~\cite{GoogleShadow} use yet a different encoding than the encoding described above, or the block-encoding in the main text. Like for the encoding in the previous paragraph, the normalization of the shadow state in Ref.~\cite{GoogleShadow} can lead to a loss of efficiency if one wishes to estimate only few entries of the correlation matrix (this loss of efficiency is avoided in our block-encoding method). In particular, the normalization of the shadow state is $a$, which is bounded as $\sqrt{{\textstyle\sum}_{j}(\langle \hat{N}_{j} \rangle-1/2)^{2}} \leq a \leq \exp(n)$, where $\langle \hat{N}_j\rangle$ is the occupation number in the mode $j$ of the represented state $\rho$.
On the other hand, when estimating densities, for example the energy density, our methods use sampling to estimate $\text{Tr}(H\rho)/K$ (with $K=\Theta(2^{n})$, the number of terms in $H$) with some error $\varepsilon$, while Ref.~\cite{GoogleShadow} estimates $\text{Tr}(H\rho)/O(2^{n/2}a)$, which, depending on the value $a$, can be more efficient. 

The precise relation between the shadow state approach \cite{GoogleShadow} and the block-encoding framework presented in this work is currently unclear. A plausible hypothesis is that the latter is strictly more powerful, due to the signal strength difference discussed above. A concrete interesting question is whether a shadow state corresponding to $M_{jk}=\mathrm{Tr}\big(\rho a^\dag_j a_k\big)$ (or, more generally, $\mathrm{Tr}\big(\rho O^\dag_j O_k\big)$) can always be produced using a block-encoding $U_M$ of $M$. In the `typical' case $\mathrm{Tr}\big(M^\dag M\big)=\Theta(2^n)$, this can be done simply by acting with $U_M$ on the maximally entangled state between $j$ and $k$ registers, and postselecting on the zero value of ancillary qubits. This `Choi–Jamiołkowski' strategy, however, does not give a constant success rate when $\mathrm{Tr}\big(M^\dag M\big)=o(2^n)$, and should be adapted.

\section{Remarks on oracle conventions}
\label{app:def}

In this work, we define row and matrix entry oracles as in Definition \ref{def:oracles}. An alternative definition of a row oracle, used in, for instance, Ref.~\cite{Gilyen_2019}, is 
\begin{equation}
    O_r^{\:\text{alt}} \ket{i}\ket{k, 0^{(n+1)-\lceil\log(s)\rceil}}=\ket{i}\ket{r(i,k)},~~\forall i\in[2^n], k\in[s],
\end{equation}
with $O_r^{\:\text{alt}}$ acting on $2(n+1)$ qubits. Again, if row $i$ contains $s'<s$ non-zero entries, then the last $n+1$ qubits are set to $\ket{1}\ket{k}$. We note that having access to $O_{r}$ in Eq.~\eqref{eq:oracles_main} implies access $O_r^{\:\text{alt}}$ and vice versa. 

In Ref.~\cite{Gilyen_2019}, $O_r^{\:\text{alt}}$ and $O_{a}$ are used to block-encode a sparse matrix $h$. In principle, this block-encoding scheme requires another (column) oracle $O_{c}^{\:\text{alt}}$ when it is used to block-encode \textit{general} sparse matrices $h$. If $h$ is also Hermitian, which is the case for all applications considered in this work, this block-encoding can be implemented with just $O_r^{\:\text{alt}}$ and $O_{a}$, since $O_{c}^{\:\text{alt}}$ can be realized using $O_r^{\:\text{alt}}$ and some SWAP gates.

\section{Margulis Expander Graphs}
\label{app:margulis}
In the main text, we have provided an example of a $d$-dimensional model which has sparse query access. Going beyond these models, in this appendix we consider an example of a model on an \textit{expander graph} which has sparse query access. Expander graphs are bounded-degree graphs, which have the so-called \textit{expansion} property. In particular, when counting the vertices away from a given vertex by a distance $d$, one obtains a number that scales exponentially with $d$. We will focus on realizing sparse access for a particular simple example, which is the Margulis expander graph.

A Margulis graph $\mathcal{G}_{M}$ of size $N^2$ has vertices $v$ labeled by tuples $v=(v_1,v_2)\in[N]\times[N]$; an edge between two vertices $u$ and $v$ is placed if $u=t_{l}(v)$ where the functions $t_l$ for $l\in [4]$ are defined as $t_0\left(\,(v_1,v_2)\,\right)=(v_1+1\,\mathrm{mod}\,N, v_2)$, $t_1\left(\,(v_1,v_2)\,\right)=(v_1, v_2+1\,\mathrm{mod}\,N)$, $t_2\left(\,(v_1,v_2)\,\right)=(v_1+v_2\,\mathrm{mod}\,N, v_2)$, and $t_3\left(\,(v_1,v_2)\,\right)=(v_1, v_2+v_1\,\mathrm{mod}\,N)$.
In other words, the first two types of edges are simple nearest-neighbour links along the vertical and horizontal directions, with periodic boundary conditions. From this perspective, the edges $t_2$ and $t_3$ are  geometrically non-local, and are the source of the expansion property of the graph.
We define our tight-binding Hamiltonian on the Margulis graph as follows. Each fermionic mode is labeled by the vertex of the graph, so the total number of modes is $N^2$. The Hamiltonian takes the form
\begin{align}
H_{\mathrm{Marg}}&=\sum_{l\in[4]}\sum_{v\in  [N]\times[N]} \left(a^\dag_{v}a^{\phantom\dagger}_{t_l(v)}+a^\dag_{t_l(v)}a^{\phantom\dagger}_{v}\right).\label{eq:Margulis_Ham}
\end{align}
For a given $v$, modular addition circuits allow to efficiently generate a list of $u=t^{\pm 1}_l(v)$. This list can be used to construct an oracle $O_r$; to ensure distinct outputs, if some of $8$ values of $u$ coincide, one stores only one of the colliding outputs. The oracle $O_a$ then represents collisions with an increased matrix element $h_{vu}$, realized by counting the times $u$ occurs in the list of $t^{\pm 1}_l(v)$. We expect that more models on expander graphs can be implemented in a similar way -- especially in the family of constant degree \href{https://en.wikipedia.org/wiki/Ramanujan_graph}{Ramanujan Cayley graphs}, of which the Margulis graph is an example.

\section{Block-encoding the thermal correlation matrix}
\label{app:fm}

In this appendix, we prove Theorem \ref{lem:fermi_dirac_main} from the main text. In particular, we prove a more detailed version of it, namely Theorem \ref{lem:fermi_dirac} below. In its proof we use Propositions \ref{prop:gilyenUh} and \ref{prop:gilyenUph} on the block-encoding of polynomials of sparse matrices, and Proposition \ref{lemma:gershgorinmain} and Lemma \ref{lem:fd_poly_apx} (of which Lemma \ref{lem:fd_poly_apx_main} is a simplified version) on constructing a polynomial approximation to our desired matrix function $M^{(\beta)}$ in Eq.~\eqref{eq:fermi_dirac_main}. We will first prove Theorem \ref{lem:fermi_dirac} and then Lemma \ref{lem:fd_poly_apx}. 

As was argued in Section \ref{sec:blockencodingmatrixfunctions} of the main text using Proposition \ref{lemma:gershgorinmain}, we wish to construct accurate polynomial approximations of $F(sx)$ for $x\in [-1,+1]$. Let us state Lemma \ref{lem:fd_poly_apx}, which will be proved at the end of this section. 
\begin{lemma}
\label{lem:fd_poly_apx}
For a function $f(x)=\frac{1}{4}\frac{1}{1+{\rm exp}\,{\beta s x}}$ (with $\beta s>0,x \in [-1,+1]$), one can efficiently construct a polynomial $p_{d}(x)$ of degree $d$ such that
\vspace{-0.7cm}
\begin{center}
\begin{align}
    & {\rm max}_{x\in [-1,+1]} |f(x)-p_{d}(x)| \nonumber \\ &\:\leq \begin{cases}
        \frac{3}{d}\big(\frac{\beta s}{\pi}\big)^{4}, & \text{ if }\frac{\beta s}{2\pi}\geq 1, \\
        \frac{10}{d}\big(\frac{\beta s}{\pi}\big)^{2}, & \text{ if }\frac{\beta s}{2\pi} < 1.
        \end{cases}
\end{align}
\end{center}
\end{lemma}

\begin{theorem}
\label{lem:fermi_dirac}
For an $s$-sparse Hamiltonian $h$ on $n$ qubits, assume access to the oracle tuple $\mathcal{O}_{h}$. We denote the controlled $(1,n+5,\varepsilon_{\text{PA}} + \varepsilon_{p(h)} + \delta)$-block-encoding of $M^{(\beta)} = \frac{1}{4}\frac{1}{1+\exp(\beta h)}$ by $C$-$U_{M^{(\beta)}}$. 
The implementation of this block-encoding requires
\begin{equation}
    \begin{cases}
        \Theta\big( \frac{\beta^4 s^4}{\varepsilon_{\text{PA}}} \big), & \text{if }\frac{\beta s}{2\pi}\geq 1, \\
        \Theta\big( \frac{\beta^2 s^2}{\varepsilon_{\text{PA}}} \big), & \text{if }\frac{\beta s}{2\pi} < 1,
    \end{cases}
\end{equation}
calls to oracles from the oracle tuple $\mathcal{O}_h$, 
\begin{align}
    O&\big( sn+n_a+\log^{5/2}(16s^{9}\beta^{8}/(\varepsilon_{\text{PA}}^{2}\varepsilon_{p(h)}^2)) \big)
\end{align}
ancillary qubits, and 
\begin{align}
    O&\big( n+(n+4)\beta^{4}s^{4}/\varepsilon_{\text{PA}}+\log^{5/2}(16s^{9}\beta^{8}/(\varepsilon_{\text{PA}}^{2}\varepsilon_{p(h)}^2)) \big)
\end{align}
additional one-qubit and two-qubit gates. To implement this block-encoding, an additional classical computing time of $\text{\normalfont{poly}}\big(\beta^{4}s^{4}/\varepsilon_{\text{PA}},\log(1/\delta)\big)$ is required. 
\end{theorem}

\begin{proof}
It follows from Proposition \ref{prop:gilyenUh} (from \cite{Gilyen_2019}) that with $O(1)$ calls to the oracle tuple $\mathcal{O}_h$, one can construct a $(s,n+3,\varepsilon_{\text{BE}_h})$-block-encoding $U_h$ of $h$ and its controlled version. For a given $\varepsilon_{\text{BE}_h}$, the required number of ancillary qubits and (additional) elementary gates are given in Proposition \ref{prop:gilyenUh}. 

Let $p_{d}(x)$ denote the degree-$d$ polynomial approximation of the function $\frac{1}{4}\frac{1}{1+\exp(\beta s x)}$ as in Lemma \ref{lem:fd_poly_apx}. It follows from Lemma~\ref{lem:fd_poly_apx} that one can efficiently construct $p_{d}$ such that 
\begin{align}
   \| p_{d}(h/s)-1/4\:M^{(\beta)} \| \leq \begin{cases}
        \frac{3}{d}\big(\frac{\beta s}{\pi}\big)^{4}, & \text{ if }\frac{\beta s}{2\pi}\geq 1, \\
        \frac{10}{d}\big(\frac{\beta s}{\pi}\big)^{2}, & \text{ if }\frac{\beta s}{2\pi} < 1.
    \end{cases}
\end{align} 
Taking $d = \Omega\big(\frac{\beta^4 s^4}{\varepsilon_{\text{PA}}}\big)$ if $\frac{\beta s}{2\pi}\geq 1$ and $d = \Omega\big(\frac{\beta^2 s^2}{\varepsilon_{\text{PA}}}\big)$ if $\frac{\beta s}{2\pi} < 1$, we achieve $\| p_{d}(h/s)-1/4\: M^{(\beta)} \| \leq \varepsilon_{\text{PA}}$. 

For $\varepsilon_{\text{PA}}<\frac{1}{4}$, we note that $|p_{d}(x)|\leq 1/2$ for $x \in [-1,+1]$. Therefore, we can apply Proposition \ref{prop:gilyenUph} (from \cite{Gilyen_2019}): A $(1,n+5,4d\sqrt{\varepsilon_{\text{BE}_h}/s}+\delta)$-block-encoding of $p_{d}(h/s)$ consists of at most $d$ uses of unitaries $U_{h}$, $U_{h}^{\dagger}$ or controlled-$U_{h}$ and $O((n+4)d)$ elementary gates. In addition, it requires a classical computation with run-time as stated in Proposition \ref{prop:gilyenUph}. We take $\varepsilon_{p(h)} := 4d\sqrt{\varepsilon_{\text{BE}_h}/s}$ so that for a given $\varepsilon_{p(h)}$, we should ensure that $\varepsilon_{\text{BE}_h} = s\varepsilon_{p(h)}^2/(16d^2)$. 

Let the $(1,n+5,\varepsilon_{p(h)} + \delta)$-block-encoding of $p_d(h/s)$ be denoted by $U_{p_d(h/s)}$. We can bound how well the block-encoding of $p_{d}(h/s)$ approximates the block-encoding of $1/4\: M^{(\beta)}$ as
\begin{align}
    \varepsilon_{\text{Tot}} = ||1/4\:M^{(\beta)} - \bra{0}^{\otimes a}\otimes \mathbbm{1}U_{p_{d}(h/s)}\ket{0}^{\otimes a}\otimes \mathbbm{1}|| \leq &\: \nonumber \\
    ||1/4\: M - p_{d}(h/s)|| \:+ \hspace{0.35cm} \nonumber \\ ||p_{d}(h/s) - \bra{0}^{\otimes a}\otimes \mathbbm{1}U_{p_{d}(h/s)}\ket{0}^{\otimes a}\otimes \mathbbm{1}|| \leq \nonumber \\ 
    \varepsilon_{\text{PA}} + \varepsilon_{p(h)} + \delta. 
    \label{eq:totalMerror}
\end{align}

We have thus constructed a $(1,n+5,\varepsilon_{\text{PA}} + \varepsilon_{p(h)} + \delta)$-block-encoding of $1/4\: M^{(\beta)}$. To implement this block-encoding, we require a number of calls to oracles from the tuple $\mathcal{O}_h$, a number of ancillary qubits, and a number of one-qubit and two-qubit gates as in the lemma statement. 
\end{proof}

\noindent
Let us now give the proof of Lemma \ref{lem:fd_poly_apx}.

\begin{proof}
For the proof of this lemma, we will employ Bernstein's theorem (Lemma \ref{theorem:Bernstein}) which bounds the error of Chebyshev approximations. Such a Chebyshev approximation of degree $d$ is of the form $p_{d}(x) = \sum_{k=0}^{d}a_{k}T_{k}(x)$, where $T_{k}(\cos(\theta)) := \cos(k\theta)$ is the degree $k$ Chebyshev polynomial of the first kind. The coefficients $a_{k}$ can be obtained by evaluating 
\begin{equation}
    a_{k} = \frac{2}{\pi}\int_{-1}^{+1}\frac{f(x)T_{k}(x)}{\sqrt{1-x^{2}}}dx,
\end{equation}
with $\frac{2}{\pi}$ replaced by $\frac{1}{\pi}$ for $k=0$. Each $a_{k}$ can be evaluated classically with $\text{poly}(\beta s k)$ resources for $f(x)$ in the lemma statement. 

Note that the function $f(z=x+i y)=\frac{1}{1+\exp(\beta s z)}$ for $\beta s> 0$ is analytic for $|y|\leq \pi/\beta s$. Hence we can pick the ellipse $E_r = \{ \frac{1}{2}(z+z^{-1})\colon |z| = r \}$ with $r = \frac{1}{2}\sqrt{(2\pi/\beta s)^{2}+4}$ on which $f(z)$ is analytic, since within this ellipse $|y|\leq \frac{\pi}{2\beta s}$. We have $|f(z)|\leq C = 1$ for $z\in E_{r}$ since for $|y|\leq \frac{\pi}{2\beta s}$, we have 
    \vspace{-0.7cm}
    \begin{center}
    \begin{align}
        |1+\exp(\beta s z)| 
        \geq &\: |1+\exp(\beta s x)\cos(\beta s y)| \geq 1.
    \end{align}
    \end{center}
   Using Lemma \ref{theorem:Bernstein}, we can thus bound ${\rm max}_{x\in [-1,+1]} |f(x)-p_{d}(x)|$ as
    \begin{equation}
        {\rm max}_{x\in [-1,+1]} |f(x)-p_{d}(x)| \leq \frac{2\big( (\pi/\beta s)^2 + 1 \big)^{-d/2}}{\frac{1}{2}\sqrt{(2\pi/\beta s)^2+4}-1}.
        \label{eq:upperboundPArough}
    \end{equation}
    Let us distinguish between scenario (1) $\beta s\geq 2\pi$ and scenario (2) $\beta s< 2\pi$. For scenario (1), we can bound
    \vspace{-0.6cm}
    \begin{center}
    \begin{align}
        \frac{1}{2}\sqrt{(2\pi/\beta s)^2+4}-1  \geq&\: \frac{1}{12}(2\pi/\beta s)^2.
    \end{align}
    \end{center}
    Furthermore, in both scenarios (1) and (2), we have that 
    \begin{equation}
        \big( (\pi/\beta s)^2 + 1 \big)^{-d/2} \leq 1/\Big((\pi/\beta s)^2 d/2 + 1\Big) \leq 1/\Big((\pi/\beta s)^2 d/2\Big). 
    \end{equation}
    Combining these two facts lead to the following bound in scenario (1)
    \begin{equation}
        {\rm max}_{x\in [-1,+1]} |f(x)-p_{d}(x)| \leq \frac{12}{d}\Big(\frac{\beta s}{\pi}\Big)^{4}.
    \end{equation}

    In scenario (2), we can simply bound the denominator in Eq.~\eqref{eq:upperboundPArough} by
    \begin{equation}
        \frac{1}{2}\sqrt{(2\pi/\beta s)^2+4}-1 \geq \frac{1}{2}\sqrt{5} - 1 \geq 1/10.
    \end{equation}
    Combining this with the upper bound above for the numerator in Eq.~\eqref{eq:upperboundPArough} (which holds in both scenarios), we obtain the following upper bound in scenario (2).
    \begin{equation}
        {\rm max}_{x\in [-1,+1]} |f(x)-p_{d}(x)| \leq \frac{40}{d}\Big(\frac{\beta s}{\pi}\Big)^{2}.
    \end{equation}
\end{proof}

\section{Block-encoding the time-evolved correlation matrix} 
\label{sec:time-evolv} 
In this appendix, we prove Theorem \ref{lem:time_evo} below, which is a generalization of Theorem \ref{lem:time_evo_main} for block-encoding $M(t_1,t_2)$ in Eq.~\eqref{eq:green2}. We will use a result from Ref. \cite{Gilyen_2018_arxiv} on block-encoding $\exp(iht)$ using a block-encoding of $h$. Note that the error of the block-encoding of $M(t_1,t_2)$ in the theorem statement accounts for potential errors in the block-encoding of the initial correlation matrix as well. 

\begin{theorem}
\label{lem:time_evo}
For an $s$-sparse Hamiltonian $h$ on $2^n$ fermionic modes, assume access to the oracle tuple $\mathcal{O}_h$. Also assume access to the $(\alpha,m,\varepsilon_M)$-block-encoding $U_M$ of a correlation matrix $M$ of a fermionic state on $2^n$ modes.  
The $\big(\alpha, 2n+m+10,\varepsilon+\varepsilon_M\big)$-block-encoding $U_{M(t_1,t_2)}$ of 
\begin{align}
M(t_1,t_2)=e^{iht_1}Me^{-iht_2},
\end{align}
can be produced using 
\begin{multline}
    D(\alpha,\varepsilon, t_1,t_2) = O\Big(s(|t_1|+|t_2|) \:+ \\ \log(12\alpha(|t_1|+|t_2|)/(|t_1|\varepsilon))+\log(12\alpha(|t_1|+|t_2|)/(|t_2|\varepsilon)) \Big)
\end{multline}
calls to oracles from the tuple $\mathcal{O}_{h}$, and a single use of the block-encoding $U_{M}$. Moreover, one uses $O\big( (n+3)(s(|t_1|+|t_2|) + \log(2\alpha(|t_1|+|t_2|)/(|t_1|\varepsilon)) + \log(2\alpha(|t_1|+|t_2|)/(|t_2|\varepsilon)) 
+ D(\alpha,\varepsilon, t_1,t_2)(n+\log^{5/2}(2\alpha s^2(|t_1|+|t_2|)/\varepsilon))\big)$ one-qubit and two-qubit gates, and $O\big(n_{a} + \log^{5/2}(2\alpha s^2(|t_1|+|t_2|)/\varepsilon)\big)$ ancillary qubits (with $n_a$ denoting the number of bits with which the entries of $h$ are specified).
\end{theorem}

\begin{proof}
A block-encoding $U_{M(t_1,t_2)}$ of $M(t_1,t_2)$ can be constructed using products of block-encodings $U_{\exp(ith)}$ of $\exp(ith)$ (for times $t_1$ and $-t_2$) and $U_{M}$ of $M$ (where the latter is a $(\alpha,m,\varepsilon_{M})$-block-encoding by assumption).

To construct a block-encoding of $\exp(iht)$, we employ a block-encoding of $h$. It follows from Proposition \ref{prop:gilyenUh} (from \cite{Gilyen_2019}) that with $O(1)$ calls to the oracle tuple $\mathcal{O}_h$, one can construct a $(s,n+3,\varepsilon_{\text{BE}_h})$-block-encoding $U_h$ of $h$ and its controlled version. For a given $\varepsilon_{\text{BE}_h}$, the required number of ancillary qubits and (additional) elementary gates are given in Proposition \ref{prop:gilyenUh}. 

Corollary 62 in \cite{Gilyen_2018_arxiv} states that to implement a $(1,n+5,|2t|\varepsilon_{\text{BE}_h})$-block-encoding of $\exp(ith)$, one is required to implement $U_h$ or $U_h^{\dagger}$ a total of $6s|t|+9\log\big((6/(|t|\varepsilon_{\text{BE}_h})\big)$ times, and controlled-$U_h$ or controlled-$U_h^{\dagger}$ three times. In addition, one has to use $O\big((n+3)(s|t|+\log((2/\varepsilon_{\text{BE}_h})\big)$ two-qubit gates and $O(1)$ ancillary qubits. So to implement the $(1,n+5,|2t|\varepsilon_{\text{BE}_h})$-block-encoding of $\exp(ith)$, one is required to call $\mathcal{O}_h$ a total of $O\big(s|t| + \log(6/(|t|\varepsilon_{\text{BE}_h})) \big)$ times. 

Using Lemma 30 in \cite{Gilyen_2019}, the block-encoding $U_{M(t_1,t_2)}$ of $M(t_1,t_2)$ can be constructed using the product $U_{M(t_1,t_2)} = (\mathbbm{1}_{n+5+m}\otimes U_{\exp(iht_1)})(\mathbbm{1}_{2n+10}\otimes U_{M})((\mathbbm{1}_{n+5+m}\otimes U_{\exp(-iht_2)})$, such that $U_{M(t_1,t_2)}$ is a $(\alpha,2n+m+10,2\alpha \varepsilon_{\text{BE}_h}(|t_1|+|t_2|) + \varepsilon_{M})$-block-encoding. To implement this product, one is thus required to make 
\begin{multline}
D(\varepsilon_{\text{BE}_h},t_1,t_2) = \\ O\big(s(|t_1|+|t_2|) + \log(6/(|t_1|\varepsilon_{\text{BE}_h}))+\log(6/(|t_2|\varepsilon_{\text{BE}_h})) \big) 
\end{multline}
calls to oracles from the tuple $\mathcal{O}_h$. In addition, one has to use a total of $O\big( (n+3)(s(|t_1|+|t_2|) + \log(1/(|t_1|\varepsilon_{\text{BE}_h}) + \log(1/(|t_2|\varepsilon_{\text{BE}_h}) 
+ D(\varepsilon_{\text{BE}_{h}}, t_1,t_2)(n+\log^{5/2}(s^2/\varepsilon_{\text{BE}_h})\big)$ one-qubit and two-qubit gates, and $O\big(n_{a} + \log^{5/2}(s^2/\varepsilon_{\text{BE}_h})\big)$ ancillary qubits. 

We stress that a controlled version $C$-$U_{M(t_1,t_2)}$ of the block-encoding of $U_{M(t_1,t_2)}$ can be implemented with equivalent resources. 
\end{proof}

\section{Block-encoding the thermal Green's function}
\label{sec:greensfunction}
In this appendix, we prove Theorem \ref{lem:greensfunction}, which is a more detailed version of Theorem \ref{lem:greensfunction_main}. In its proof we will again use Propositions \ref{prop:gilyenUh} and \ref{prop:gilyenUph} on the block-encoding of polynomials of sparse matrices. In addition, we will use Proposition \ref{lemma:gershgorinmain} and Lemma \ref{lem:greens_poly_apx} (of which Lemma \ref{lem:greens_poly_apx_main} is a simplified version) on constructing a polynomial approximation to our desired matrix function $G^{(\delta,\beta,\omega)}$ in Eq.~\eqref{eq:greensomega_main}. 

As was argued in Section \ref{sec:blockencodingmatrixfunctions} using Proposition \ref{lemma:gershgorinmain}, we would like to construct accurate polynomial approximations of $F(sx)$ for $x\in [-1,+1]$. The function to be approximated for block-encoding $G^{(\delta,\beta,\omega)}$ is
\begin{multline}
    g^{(\delta,\beta,\omega)}(x) := \frac{\delta}{8}\bigg[\Big( 1-\frac{1}{1+\exp(\beta sx)} \Big) \frac{1}{i\delta - (sx+\omega)} \\ + \Big( \frac{1}{1+\exp(\beta sx)} \Big) \frac{-1}{i\delta + (sx+\omega)} \bigg].
\label{eq:greensfunction}
\end{multline}
Note that $g^{(\delta,\beta,\omega)}(z)$ ($z\in \mathbb{C}$) has poles at $z = \frac{i\delta - \omega}{s}$ and $z = \frac{-i\delta - \omega}{s}$; the regularization parameter $\delta$ ensures that these poles lie off the real axis. For convenience, we define the functions $g_{1}^{(\delta,\omega)}(z) = 1/(i\delta - (sz+\omega))$ and $g_{2}^{(\delta,\omega)}(z) = -1/(i\delta + (sz+\omega))$. Due to the poles, $|g_{1,2}^{(\delta,\omega)}(x)|$ can still grow as $1/\delta$. To be able to apply Proposition \ref{prop:gilyenUph} for block-encoding polynomials, we have to ensure that the polynomial that approximates $g^{(\delta,\beta,\omega)}(x)$ has absolute value at most $1/2$ for $x\in [-1,+1]$. That is the reason for including a factor of $\delta/8$ in $g^{(\delta,\beta,\omega)}(x)$ (so that its absolute value is at most $1/4$, and that of its polynomial approximation at most $1/2$ for approximation error less than $1/4$). 

Let us first state the following lemma, the proof of which will be provided at the end of this section, which will be used in the proof of Theorem \ref{lem:greensfunction} (and thus Theorem \ref{lem:greensfunction_main}) on the block encoding of the matrix $G^{(\delta,\beta,\omega)}(h)$. 

\begin{lemma}
\label{lem:greens_poly_apx}
For a function $g^{(\delta,\beta,\omega)}(x)$ as in Eq.~ \eqref{eq:greensfunction} (with $\beta,\delta,s>0$ and $x \in [-1,+1]$), one can efficiently construct a polynomial $p_{d}(x)$ of (even) degree $d$ such that
\vspace{-0.7cm}
\begin{center}
\begin{align}
    & {\rm max}_{x\in [-1,+1]} | g^{(\delta,\beta,\omega)}(x)-p_{d}(x)| \nonumber \\ &\:\leq \begin{cases}
        \frac{12}{d}\big(\frac{\beta s}{\pi}\big)^{4}, & \text{ if }\frac{\beta s}{2\pi}\geq 1, \\
        \frac{40}{d}\big(\frac{\beta s}{\pi}\big)^{2}, & \text{ if }\frac{\beta s}{2\pi} < 1.
        \end{cases}
        \nonumber \\&\:\hspace{0.4cm}+
        \begin{cases}
        \frac{128}{d}\big(\frac{s}{\delta}\big)^{4}, & \text{ if }\frac{2s}{\delta}\geq 1, \\
        \frac{32}{d}\big(\frac{s}{\delta}\big)^{2}, & \text{ if }\frac{2s}{\delta} < 1.
        \end{cases}
\label{eq:greenspolyapprox}
\end{align}
\end{center}
\end{lemma}

\begin{theorem}
\label{lem:greensfunction}
For an $s$-sparse Hamiltonian $h$ on $n$ qubits, assume access to the oracle tuple $\mathcal{O}_{h}$. We denote the controlled $(1,n+5, \varepsilon_{\text{PA}} + \varepsilon_{p(h)} + \delta_{\text{class}})$-block-encoding of $\frac{1}{4}G^{(\delta,\beta,\omega)}(h)$ in Eq.~\eqref{eq:greensomega_main} by $C$-$U_{G^{(\delta,\beta,\omega)}}$. 
The implementation of this block-encoding requires
\begin{equation}
    \begin{cases}
        \Theta\Big( \frac{(\beta s)^4}{\varepsilon_{\text{PA}}} \Big), & \text{if }\frac{\beta s}{2\pi}\geq 1, \\
        \Theta\Big( \frac{(\beta s)^2}{\varepsilon_{\text{PA}}} \Big), & \text{if }\frac{\beta s}{2\pi} < 1.
        \end{cases}+
        \begin{cases}
        \Theta\Big( \frac{s^4}{\delta^{4}\varepsilon_{\text{PA}}} \Big), & \text{if }\frac{2s}{\delta}\geq 1, \\
        \Theta\Big( \frac{s^2}{\delta^{2}\varepsilon_{\text{PA}}} \Big), & \text{if }\frac{2s}{\delta} < 1.
        \end{cases}
\end{equation}
calls to oracles from the oracle tuple $\mathcal{O}_h$,
\begin{align}
    O&\big( sn+n_a+\log^{5/2}\big(16s^9(\beta^{4}+1/\delta^4)^2/(\varepsilon_{\text{PA}}^{2}\varepsilon_{p(h)}^2)\big) \big) 
\end{align}
ancillary qubits, and 
\begin{align}
    O&\big( n+(n+4)(\beta^{4}s^{4}+s^4/\delta^4)/\varepsilon_{\text{PA}}\\ \nonumber &\hspace{1.7cm}+\log^{5/2}\big(16s^9(\beta^{4}+1/\delta^4)^2/(\varepsilon_{\text{PA}}^{2}\varepsilon_{p(h)}^2)\big) \big)
\end{align}
additional one-qubit and two-qubit gates. To implement this block-encoding, an additional classical computing time of $\text{\normalfont{poly}}\big((\beta^{4}s^{4}+s^4/\delta^4)/\varepsilon_{\text{PA}},\log(1/\delta_{\text{class}})\big)$ is required. 
\end{theorem}
\begin{proof}
Like in the proof of Theorem \ref{lem:fermi_dirac}, we employ Proposition \ref{prop:gilyenUh} (from \cite{Gilyen_2018_arxiv}) to construct a $(s,n+3,\varepsilon_{\text{BE}_{h}})$-block-encoding $U_{h}$ of $h$. Using this block encoding, we construct a block encoding of a polynomial approximation of $G^{(\delta,\beta,\omega)}(h)$. Let $p_{d}(x)$ denote the degree-$d$ polynomial approximation of the function $g^{(\delta,\beta,\omega)}(x)$ from Lemma \ref{lem:greens_poly_apx}. It follows from Lemma \ref{lem:greens_poly_apx} that one can efficiently construct $p_{d}$ such that
\begin{equation}
    \bigl\lvert\bigl\lvert p_{d}(h/s) - G^{(\delta,\beta)}(\omega,h) \bigr\rvert\bigr\rvert
\end{equation}
is upper bounded by the RHS of the inequality in Eq.~\eqref{eq:greenspolyapprox}. Hence, taking 
\begin{align}
    & d \leq \begin{cases}
        \Theta\Big( \frac{(\beta s)^4}{\varepsilon_{\text{PA}}} \Big), & \text{if }\frac{\beta s}{2\pi}\geq 1, \\
        \Theta\Big( \frac{(\beta s)^2}{\varepsilon_{\text{PA}}} \Big), & \text{if }\frac{\beta s}{2\pi} < 1.
        \end{cases}+
        \begin{cases}
        \Theta\Big( \frac{s^4}{\delta^{4}\varepsilon_{\text{PA}}} \Big), & \text{if }\frac{2s}{\delta}\geq 1, \\
        \Theta\Big( \frac{s^2}{\delta^{2}\varepsilon_{\text{PA}}} \Big), & \text{if }\frac{2s}{\delta} < 1.
        \end{cases}
\end{align}
we obtain $\bigl\lvert\bigl\lvert p_{d}(h/s) - \delta/8\:G^{(\delta,\beta)}(\omega,h) \bigr\rvert\bigr\rvert \leq \varepsilon_{\text{PA}}$. 

For $\varepsilon_{\text{PA}}\leq \frac{1}{4}$, we note that $|p_{d}(x)|\leq 1/2$ for $x\in [-1,+1]$, allowing us to apply Proposition \ref{prop:gilyenUph} (from \cite{Gilyen_2019}). A $(1,n+5,4d\sqrt{\varepsilon_{\text{BE}_h}/s}+\delta)$-block-encoding of $p_{d}(h/s)$ consists of a circuit with $O((n+4)d)$ one-qubit and two-qubit gates, and at most $d$ calls to unitaries $U_{h}$, $U_{h}^{\dagger}$ or controlled-$U_{h}$. The classical description of this circuit can be classically computed in $O\big(\text{poly}(d,\log(1/\delta_{\text{class}}))\big)$ time. We define $\varepsilon_{p(h)} := 4d\sqrt{\varepsilon_{\text{BE}_h}/s}$ so that for a given $\varepsilon_{p(h)}$, we should ensure that $\varepsilon_{\text{BE}_h} = s\varepsilon_{p(h)}^2/(16d^2)$. 

Let the $(1,n+5,\varepsilon_{p(h)} + \delta_{\text{classical}})$-block-encoding of $p_d(h/s)$ be denoted by $U_{p_d(h/s)}$. Like in the proof of Theorem \ref{lem:fermi_dirac}, we have that $\varepsilon_{\text{Tot}} = \|G^{(\delta,\beta,\omega)}(h) - \bra{0}^{\otimes a}\otimes \mathbbm{1}U_{p_{d}(h/s)}\ket{0}^{\otimes a}\otimes \mathbbm{1}\| \leq \varepsilon_{\text{PA}} + \varepsilon_{p(h)} + \delta_{\text{class}}$. We have thus constructed a $(1,n+5,\varepsilon_{\text{PA}} + \varepsilon_{p(h)} + \delta_{\text{class}})$-block-encoding of $G^{(\delta,\beta,\omega)}(h)$. To implement this block-encoding, we require a number of calls to oracles from the tuple $\mathcal{O}_h$, a number of ancillary qubits, and a number of one-qubit and two-qubit gates as in the lemma statement. 
\end{proof} 

\noindent
Let us now give the proof of Lemma \ref{lem:greens_poly_apx}. 
\begin{proof}
We wish to approximate $g^{(\delta,\beta,\omega)}(x)$ in Eq. \eqref{eq:greenspolyapprox} by a polynomial of degree $d$. Let us first express $g^{(\delta,\beta,\omega)}(x)$ as 
\begin{equation}
    \delta/8\Big((1-f^{(\beta)}(x))g_{1}^{(\delta,\omega)}(x) + f^{(\beta)}(x)g_{2}^{(\delta,\omega)}(x)\Big),
\end{equation}
and its degree-$d$ polynomial approximation $p_{d}(x)$ by 
\begin{equation}
    \delta/8\Big(\big(1-f^{(\beta)}_{d/2}(x)\big)g_{1,d/2}^{(\delta,\omega)}(x) + f^{(\beta)}_{d/2}(x)g_{2,d/2}^{(\delta,\omega)}(x)\Big).
\end{equation}
Note that 
\begin{multline}
    |\delta/8\:g^{(\delta,\beta,\omega)}(x) - p_{d}(x)| \leq \delta/8\Big( |g_{1}^{(\delta,\omega)}(x) - g_{1,d/2}^{(\delta,\omega)}(x)| \\ + |g_{2}^{(\delta,\omega)}(x) - g_{2,d/2}^{(\delta,\omega)}(x)| \Big) + 1/2|f^{(\beta)}(x) - f^{(\beta)}_{d/2}(x)|,
\end{multline}
where we have used that $|g_{1,d/2}^{(\delta,\omega)}(x)|,|g_{2,d/2}^{(\delta,\omega)}(x)|\leq 2/\delta$ for sufficiently large $d$ (note that $|g_{1}^{(\delta,\omega)}(x)|,|g_{2}^{(\delta,\omega)}(x)|\leq 1/\delta$).
Using the bound on $\max_{x\in [-1,+1]}|f^{(\beta)}(x) - f^{(\beta)}_{d/2}(x)|$ from Lemma \ref{lem:fd_poly_apx}, and applying Bernstein's theorem \cite{Bernstein} (i.e., Lemma \ref{theorem:Bernstein}) to the functions $g_{1}^{(\delta,\omega)}(x)$ and $g_{2}^{(\delta,\omega)}(x)$ (with a Bernstein ellipse $E_{r}$ with $r = \sqrt{(\delta/(2s))^2+1}$), we obtain the upper bound on $\max_{x\in [-1,+1]}|g^{(\delta,\beta,\omega)}(x) - p_{d}(x)|$ in the lemma statement. 
\end{proof}


\section{Proof of Lemma \ref{lem:obs_extraction}}
\label{sec:cor-estim}

\begin{proof}[Proof of Lemma \ref{lem:obs_extraction}]
By assumption, we have that $\bigl\lvert\bra{i}{A}\ket{j}-\alpha \bra{0}^{\otimes m}\bra{i}  U_{A} \ket{0}^{\otimes m} \ket{j} \bigr\rvert \leq \varepsilon$, where $U_A$ acts on $n+m$ qubits. Let us consider estimating $\bra{0}^{\otimes m}\bra{i}  U_{A} \ket{0}^{\otimes m} \ket{j}$, which can alternatively be expressed as 
\begin{equation}
\bra{0}^{\otimes m} \bra{0}^{\otimes n} (\mathbbm{1} \otimes U_{i}^{\dagger}) \:U_{A}\: (\mathbbm{1} \otimes U_{j}) \ket{0}^{\otimes m} \ket{0}^{\otimes n},
\end{equation}
where $U_{i},U_j$ are depth-1 circuits which prepare bit-strings $i$ and $j$. We denote the estimate of $\bra{0}^{\otimes m}\bra{i}  U_{A} \ket{0}^{\otimes m} \ket{j}$ by $\widetilde{\bra{i}A\ket{j}}$, so that if
$\bigl\lvert \bra{0}^{\otimes m}\bra{i}  U_{A} \ket{0}^{\otimes m} \ket{j} - \widetilde{\bra{i}A\ket{j}} \bigr\rvert \leq \tilde{\varepsilon}$, then $\bigl\lvert \bra{i} A \ket{j} - \alpha \widetilde{\bra{i}A\ket{j}} \bigr\rvert \leq \varepsilon + \alpha \tilde{\varepsilon}$. 

One can obtain the estimate $\widetilde{\bra{i}A\ket{j}}$ by running a series of Hadamard test circuits on $n+m+1$ qubits. These circuits correspond to running
\vspace{-0.6cm}
\begin{center}
\begin{align}
    \big(\mathbbm{1}\otimes [H\:R_{z}(\theta)]_{\text{a}}\big) \big( \mathbbm{1}\otimes \ket{0}\bra{0}_{\text{a}} + U\otimes \ket{1}\bra{1}_{\text{a}} \big)\big(\mathbbm{1}\otimes H_{\text{a}}\big),
\end{align}
\end{center}
where $U = (U_{i}^{\dagger}\otimes \mathbbm{1}) \:U_{A}\: (U_{j}\otimes \mathbbm{1})$, on the state $\ket{0}^{\otimes m}\ket{0}^{\otimes n}\ket{0}_{\text{a}}$ (with the final qubit being an ancillary qubit). The output state of the ancillary qubit is measured a total of $D(\tilde{\varepsilon},\delta)$ times, half of the times for $\theta = 0$ and half of the times for $\theta = \pi/2$. The fractions of output-$0$ measurements for $\theta = 0$ and $\theta = \pi/2$ provide estimates of $\frac{1}{2}+\frac{1}{2}\text{Re}\big( \bra{0}^{\otimes m}\bra{i}  U_{A} \ket{0}^{\otimes m} \ket{j} \big)$ and $\frac{1}{2}-\frac{1}{2}\text{Im}\big( \bra{0}^{\otimes m}\bra{i}  U_{A} \ket{0}^{\otimes m} \ket{j} \big)$, respectively. Using a Chernoff concentration bound, one can show that $\bigl\lvert \widetilde{\bra{i}A\ket{j}} - \bra{0}^{\otimes m}\bra{i}  U_{A} \ket{0}^{\otimes m} \ket{j} \bigr\rvert \leq \tilde{\varepsilon}$ with probability at least $1-\delta$ for $D(\tilde{\varepsilon},\delta) = \Theta\big(\tilde{\varepsilon}^{-2}\log(4\delta^{-1})\big)$. 

One can thus obtain an estimate of $\bra{i}A\ket{j}$ (given by $\alpha \widetilde{\bra{i}A\ket{j}}$) up to error $\varepsilon + \alpha\tilde{\varepsilon}$ with probability $1-\delta$, using $D(\tilde{\varepsilon},\delta) = \Theta\big(\tilde{\varepsilon}^{-2}\log(4\delta^{-1})\big)$ calls to $C\text{-}U_A$. 

\end{proof}

\noindent

\section{BQP-completeness}
\label{sec:bqp}
Here we prove Theorem \ref{thm:time_evo_main} in the main text, using the next Lemma \ref{lem:hopping_overlap} as a small tool:
\begin{proof}[Proof of Theorem \ref{thm:time_evo_main}]
It is straightforward to see that evaluating the matrix element $M_{jj}(t)$ of the correlation matrix $M(t)=e^{iht}M_0 e^{-iht}$ at $t = \text{poly}(n)$ is a problem in BQP, given the promise. By Lemmas~\ref{lem:obs_extraction} and \ref{lem:time_evo}, given access to $\mathcal{O}_{M_0}$ and $\mathcal{O}_h$ as poly$(n)$-sized quantum circuits, the problem is solved with $\text{poly}(n)$ quantum effort.

To show BQP-hardness of our problem, we use the fact that for any promise problem in BQP of problem size $m$, we have the following property \cite{BV}: the problem can be decided by acting on an $k={\rm poly}(m)$-qubit input $\ket{00\ldots 0}$ with (a uniform family of) ${\rm poly}(k) = {\rm poly}(m)$-sized quantum circuits, outputing $1$ (on the first qubit) with probability at least $2/3$ in case ${\rm YES}$, and 1 with probability at most $1/3$ in case ${\rm NO}$. In addition, one can boost the success and failure probabilities $2/3 \rightarrow 1-\exp(-\Theta(k))$ and $1/3 \rightarrow \exp(-\Theta(k))$, by running $k$ instances of the ${\rm poly}(k)$-sized circuits in parallel and taking a majority vote on the first qubit of the output state for each instance (and copying the answer onto an ancillary qubit). The circuit corresponding to this boosted scenario acts on $q = k^2$ qubits, and its success and failure probabilities are respectively $1-\exp(-\Theta(\sqrt{q}))$ and $\exp(-\Theta(\sqrt{q}))$. Let the quantum circuit for this problem with boosted probabilities be 
\begin{align}
    U=W_L \ldots W_1,
\end{align}
where $W_l$ are elementary one-qubit and two-qubit gates and $L={\rm poly}(k) = {\rm poly}(\sqrt{q})$. We represent this decision problem using time-evolution with a sparse circuit Hamiltonian. The circuit Hamiltonian, acting on a $q_{\text{clock}} = \log_{2}(L+1)$-qubit clock space (we assume wlog that $\log_{2}(L+1)$ is an integer) and the $q$-qubit space is given by 
\begin{equation}
\label{eq:hist_Hamiltonian}
    h=\sum^L_{l=1} \left(\ket{l+1}\bra{l}_{\rm clock} \otimes W_l+\ket{l}\bra{l+1}_{\rm clock} \otimes W_l^{\dag}\right).
\end{equation}
We take $n = q_{\text{clock}}+q$ and note that $q_{\text{clock}}<q$ for sufficiently large $q$, so that $n/2\leq q\leq n$.
The matrices $W_{l}$ have at most $4$ non-zero entries in a given row/column. Therefore, $h$ is at most $8$-sparse. Since $\{W_{l}\}_{l=1}^{L}$ are unitary matrices, the entries of $h$ are $O(1)$ in absolute value. 

Consider the evolution $\ket{\psi(t)}=e^{-iht}\ket{1}_{\rm clock}\ket{00 \ldots 0}$ with the Hamiltonian $h$ from Eq.~\eqref{eq:hist_Hamiltonian}. This state can be decomposed as
    \begin{align}
        \ket{\psi(t)}=\sum^{L+1}_{l=1} \alpha_{l,t} \ket{l}_{\rm clock}\otimes \prod^{l-1}_{l'=1} W_{l'} \ket{00\ldots 0}
    \end{align}
    with coefficients $\alpha_{l,t}$ given by
    \begin{align}
        \sum^{L+1}_{l=1} \alpha_{l,t} \ket{l}\equiv e^{-iJt}\ket{1}_{\rm clock}, 
    \end{align}
    where $J$ is a Hamiltonian on the clock register
    \begin{align}
        J=\sum^L_{l=1} \left(\ket{l+1}\bra{l}_{\rm clock} +\ket{l}\bra{l+1}_{\rm clock} \right).
    \end{align}
 Given the encoding of the clock register, one can write the probability of measuring $\ket{L+1}_{\text{clock}}$ on the clock and measuring $\ket{1}$ on the first of the $q$ qubits as
     \begin{align}
         p \equiv
         \bigl\lvert(\langle L+1|_{\rm clock} \otimes \langle 1|_1) \: |\psi(t)\rangle\bigr\rvert^2= \notag \\ \bra{1}_{\rm clock}\bra{00 \ldots 0}e^{iht} M_0 e^{-iht}\ket{1}_{\rm clock}\ket{00 \ldots 0},
         \label{eq:defMtilde}
     \end{align}
     with $M_0=\frac{1}{2^{q_{\text{clock}}+1}}\prod_{j=1}^{q_{\text{clock}}}(\mathbbm{1}-Z_{{\rm clock},j})(\mathbbm{1}-Z_{{\rm qubit},1})$.
    Hence, when the state $U\ket{00\ldots 0}$ outputs 1 on the first qubit with probability at least $1-\exp(-\sqrt{q})$ (YES), it follows through Lemma \ref{lem:hopping_overlap} that $p = \Omega(1/\text{poly}(\sqrt{q})) = \Omega(1/\text{poly}(\sqrt{n}))$.  When the state $U\ket{00\ldots 0}$ outputs 1 on the first qubit with probability at most $\exp(-\sqrt{q})$ (NO), then $p\leq \exp(-\sqrt{q}) \leq \exp(-\sqrt{n/2})$ through Lemma \ref{lem:hopping_overlap}. Now, observe that $M_0$ is a valid and sparse correlation matrix of a multi-particle free-fermionic state on $2^{n}$ modes (in particular, a fraction $\Theta\big(1/\text{poly}(\sqrt{n})\big)$ of the modes is occupied), which is evolved in time $t={\rm poly}(\sqrt{n})$ by the sparse Hamiltonian $h$, after which one wishes to estimate a particular matrix element (labeled, say, by $j=1_{\rm clock},00 \ldots 0$) of the time-evolved matrix, which is the problem stated in Theorem \ref{thm:time_evo_main}. 
    The only thing left to argue is that given the description of $\{W_l\}$, one can implement $\mathcal{O}_h$ in Definition \ref{def:oracles} as a ${\rm poly}(n)$-sized circuit.

    \bigbreak
    \textbf{\textit{Oracle implementation:}}
    The oracle $O_{r}$ from Definition \ref{def:oracles}, acting on $(s+1) (q_{\text{clock}}+q+1)$ qubits, can be implemented as follows. For convenience, we label the first $(q_{\text{clock}}+q+1)$ qubits by $A$ and the last $s$ $(q_{\text{clock}}+q+1)$-qubit registers by $B_1,\ldots,B_s$. For simplicity and wlog, we assume that all $W_l$ are two-qubit gates and all entries of $W_{l}$ in their two-qubit sub-spaces are non-zero. Note that for each $l\in \{1,2,\ldots,L\}$, we have access to the labels $Q_{1}^{(l)}$ and $Q_{2}^{(l)}$ (with $Q_{1}^{(l)} < Q_{2}^{(l)}$) of the qubits on which $W_{l}$ acts non-trivially. The structure of $h$ is such that each row contains $8$ non-zero entries (apart from the rows associated with clock states $\ket{1}_{\text{clock}}$ and $\ket{L+1}_{\text{clock}}$), with a row $\ket{i} = \ket{l}_{\text{clock}}\ket{x}$ having four non-zero entries associated with clock register state $\ket{l-1}_{\text{clock}}$ and four non-zero entries associated with clock register state $\ket{l+1}_{\text{clock}}$. These entries correspond to the entries $\bra{x_{Q_{1}^{l-1}},x_{Q_{2}^{l-1}}}W_{l-1}\ket{y_{1},y_{2}}$ and $\bra{x_{Q_{1}^{l}},x_{Q_{2}^{l}}}W_{l}\ket{y_{1},y_{2}}$ (for $y\in \{0,1\}^{2}$), respectively. The rows associated with clock states $\ket{1}_{\text{clock}}$ and $\ket{L+1}_{\text{clock}}$ are $4$-sparse.

    We take workspace in the form of $2(L+1)$ additional $(q_{\text{clock}}+q)$-qubit registers (initialized in $\ket{00\ldots 0}$), denoted by $C_{1},\ldots,C_{2(L+1)}$. For each $j\in\{1,2,\ldots,L+1\}$, we transform the first $(L+1)$ qubits on registers $C_{2j-1}$ and $C_{2j}$ to $\ket{j}_{\text{clock}}$. Then, for each $j\in\{2,3,\ldots,L\}$ (so excluding $1$ and $L+1$), we flip qubits $q_{\text{clock}}+Q_{1}^{j-1}$ and $q_{\text{clock}}+Q_{2}^{j-1}$ on register $C_{2j-1}$ and qubits $q_{\text{clock}}+Q_{1}^{j}$ and $q_{\text{clock}}+Q_{2}^{j}$ on register $C_{2j}$ to $\ket{1}$. In addition, we flip qubits $q_{\text{clock}}+Q_{1}^{1}$ and $q_{\text{clock}}+Q_{2}^{1}$ on register $C_{2}$ and $q_{\text{clock}}+Q_{1}^{L}$ and $q_{\text{clock}}+Q_{2}^{L}$ on register $C_{2L-1}$ to $\ket{1}$. 

    Controlled on the clock state on register $A$ being $\ket{l}_{\text{clock}}$, we set the clock state to $\ket{l-1}_{\text{clock}}$ on registers $B_{1},\ldots,B_{4}$ (provided that $l>1$) and to $\ket{l+1}_{\text{clock}}$ on register $B_{5},\ldots,B_{8}$ (provided that $l<L+1$). Controlled on the last $q$ qubits of register $A$ being in state $\ket{x}$, we copy $\ket{x}$ onto the final $q$ qubits of $B_{1},\ldots,B_{4}$, excluding qubits $q_{\text{clock}}+Q_{1}^{l-1}$ and $q_{\text{clock}}+Q_{2}^{l-1}$. These latter two qubits are transformed to $\ket{00}$, $\ket{01}$, $\ket{10}$ and $\ket{11}$ on registers $B_{1},\ldots,B_{4}$, respectively. Similarly, we copy $\ket{x}$ onto the final $q$ qubits of $B_{5},\ldots,B_{8}$, apart from qubits $q_{\text{clock}}+Q_{1}^{l}$ and $q_{\text{clock}}+Q_{2}^{l}$, which are respectively transformed to  $\ket{00}$, $\ket{01}$, $\ket{10}$ and $\ket{11}$. These operations make use of the states in the workspace registers $C_{1},\ldots,C_{2(L+1)}$, which are uncomputed at the end of the protocol. In accordance with Definition \ref{def:oracles}, we need to account for rows of $h$ having less than $8$ non-zero entries. Since the rows of $h$ associated with clock states $\ket{1}_{\text{clock}}$ and $\ket{L+1}_{\text{clock}}$ are $4$-sparse, registers $B_{1},\ldots,B_{4}$ are set to resp. $\ket{1}\otimes \ket{5}_{q_{\text{clock}}+q},\ldots,\ket{1}\otimes \ket{8}_{q_{\text{clock}}+q}$ controlled on the $A$ clock state being $\ket{1}_{\text{clock}}$ (after which registers $(B_{1},\ldots,B_{4})$ and $(B_{5},\ldots,B_{8})$ are swapped), and registers $B_{5},\ldots,B_{8}$ are set to resp. $\ket{1}\otimes \ket{5}_{q_{\text{clock}}+q},\ldots,\ket{1}\otimes \ket{8}_{q_{\text{clock}}+q}$ controlled on the $A$ clock state being $\ket{L+1}_{\text{clock}}$. The size of the circuit implementing $O_{r}$ is ${\rm poly}(n)$. 

    To implement oracle $O_a$, let us note that wlog the entries of $W_{l}$ are $0$, $\pm 1/\sqrt{2}$ or $1$, so that the entries can be encoded into a three bit string. By employing additional $\text{poly}(n)$-sized workspace (note that $L=\text{poly}(\sqrt{q})$ and each $W_{l}$ has $16$ entries), the oracle $O_{a}$ can be implemented (by a ${\rm poly}(n)$-sized circuit). 

\end{proof}

{\em Remark}: Like in \cite{babbush2023exponential}, we could have adapted the BQP-verification circuit to output the state $\ket{0}_a \otimes \ket{00\ldots 0}$ (so all qubits back to their initial state and an additional ancilla qubit $a$ to 0) with high probability in the NO case, and with low probability in the YES case. This is done by simply copying the answer of the BQP-circuit onto an additional ancilla qubit $a$ and applying the gates $W_L\ldots W_1$ in reverse on the other qubits. If we use this cleaned-up circuit, it means that we are interested in estimating the probability for a specific output state --- all qubits in $\ket{0}$ and clock state in $\ket{L+1}_{\rm clock}$ --- and this corresponds to estimating an entry of a time-evolved rank-1 projector $\tilde{M}_0$, corresponding to a single-particle state. 
Hence not surprisingly, time-evolution of single-particle states is also BQP-complete, as was shown in Theorem 3 in \cite{babbush2023exponential} (where more work was done to bring $h$ in sign-free form to directly correspond to a sum of kinetic and potential energy). \\
  
 The following lemma, which is used in the proof of Theorem \ref{thm:time_evo_main}, mainly follows the approach of \cite{babbush2023exponential}. Instead of employing this lemma, one could also adapt the coefficients in the hopping Hamiltonian $h$ in Eq.~\eqref{eq:hist_Hamiltonian} to allow for a perfect 1D state transfer from $\ket{1}_{\rm clock}\rightarrow \ket{L+1}_{\rm clock}$, using an idea first suggested by Peres \cite{peres}, see also \cite{babbush2023exponential}: such adaptation requires extra ancilla qubit overhead in realizing the time-dynamics of $h$, hence we omit it.

\begin{lemma}
\label{lem:hopping_overlap}
    For a Hamiltonian $J=\sum^{L}_{l=1} (\ket{l}\bra{l+1}+\ket{l+1}\bra{l})$ on a $(L+1)$-dim Hilbert space with basis states $\ket{l},~l\in\{1,\ldots, L+1\}$, there exists a $t=O(L^2 \log L)$ such that
    \begin{align}
        |\bra{L+1} e^{-iJt} \ket{1}|= \Omega(1/\sqrt{L}).
    \end{align}
\end{lemma}

\begin{proof}
The Hamiltonian $J$ has eigenstates
\begin{equation}
    \ket{\psi_{k}}=\sum_{j=1}^{L+1}\alpha^{(k)}_{j}\ket{j},\text{ with }\alpha_{j}^{(k)} = \sqrt{\frac{2}{L+2}}\sin\Big( \frac{\pi jk}{L+2} \Big),
\end{equation}
and eigenvalues 
\begin{equation}
    \epsilon_{k} = 2\cos\Big( \frac{\pi k}{L+2} \Big),
\end{equation}
with $k=1 \ldots L+1$. We note that the gap between any two eigenvalues is at most $4$. To prove a lower bound on $|\bra{L+1}e^{-iJt}\ket{1}|$, we will derive a lower bound on the gaps $\Delta_{m} := |\epsilon_{m+1}-\epsilon_{m}|$ (for $m=1,2 \ldots L$) between the eigenvalues of $J$:
\begin{align}
    \Delta_{m} = |\epsilon_{m+1}-\epsilon_{m}| \geq \nonumber \\ \frac{\pi}{L+2} \min_{x\in \big[ \frac{m\pi}{L+2},\frac{(m+1)\pi}{L+2} \big]} \Bigl\lvert \frac{d \,2\cos(x)}{dx} \Bigr\rvert \geq \nonumber \\ \frac{2\pi}{L+2} \sin\Big( \frac{\pi}{L+2} \Big)  = \Omega(1/(L+2)^{2}).
\end{align}

Using the eigendecomposition of $J$, we infer that
\begin{equation}
    \bra{L+1}e^{-iJt}\ket{1} = \frac{2}{L+2}\sum_{k=1}^{L+1} e^{-i\epsilon_{k} t}(-1)^{k-1}\sin^{2}\Big( \frac{\pi k}{L+2} \Big),
\end{equation}
so that 
\begin{multline}
    |\bra{L+1}e^{-iJt}\ket{1}|^{2} = \Big( \frac{2}{L+2} \Big)^{2} \times \\ \sum_{k,k' = 1}^{L+1}e^{-i(\epsilon_{k} - \epsilon_{k'})t}(-1)^{k+k'}\sin^{2}\Big( \frac{\pi k}{L+2} \Big)\sin^{2}\Big( \frac{\pi k'}{L+2} \Big).
\end{multline}

To show that there must be a time $t$ for which $|\bra{L+1}e^{-iJt}\ket{1}|^{2} = \Omega(1/L)$, we use the fact that a probabilistically chosen time in a sufficiently large interval will give high success probability \cite{nagaj:phd}, and hence there must exist a specific time which works sufficiently well. More precisely,
for $k\neq k'$, there must exist a probability distribution $\{p(t)\}_{t=0}^{T}\geq 0$, $\sum_{t=0}^{T}p(t) = 1$, such that
\begin{equation}
    \biggl\lvert \sum_{t=0}^{T} p(t) e^{-i(\epsilon_{k} - \epsilon_{k'})t} \biggr\rvert \leq \varepsilon,
\end{equation}
provided that $\Delta = \Omega\big(1/(L+2)^{2}\big)$ and $T = O\big((L+2)^2\:\log(1/\varepsilon)\big)$. Examples of probability distributions for which this is true are given in Ref. \cite{phaserandomization}. 

Therefore, for those $\{p(t)\}$'s we have that 
\begin{align}
    \biggl\lvert \sum_{k\neq k'}\sum_{t=0}^{T} p(t)e^{-i(\epsilon_{k} - \epsilon_{k'})t}(-1)^{k+k'} \times &\: \nonumber \\ \sin^{2}\Big( \frac{\pi k}{L+2} \Big)&\: \sin^{2}\Big( \frac{\pi k'}{L+2} \Big) \biggr\rvert \leq \nonumber \\ \varepsilon \sum_{k\neq k'} \sin^{2}\Big( \frac{\pi k}{L+2} & \Big) \sin^{2}\Big( \frac{\pi k'}{L+2} \Big) = \nonumber \\ \varepsilon \Big( \frac{(L+2)^2}{4} - \frac{3(L+2)}{8} & \Big) \: \leq \: \varepsilon \frac{(L+2)^2}{4},
\end{align}
where the equality follows from direct computation. We thus conclude that 
\begin{multline}
    \biggl\lvert \sum_{t=0}^{T} p(t)|\bra{L+1}e^{-iJt}\ket{1}|^{2} - \\ \sum_{t=0}^{T}p(t)\Big(\frac{2}{L+2}\Big)^2\sum_{k=1}^{L+1}\sin^4\Big( \frac{\pi k}{L+2} \Big)\biggr\rvert \leq \varepsilon. 
\end{multline}
The term $\sum_{t=0}^{T}p(t)\big(\frac{2}{L+2}\big)^2\sum_{k=1}^{L+1}\sin^4\big( \frac{\pi k}{L+2} \big)$ can be evaluated to be $\frac{3}{2(L+2)}$. So choosing, for instance, $\varepsilon = \frac{1}{2(L+2)}$, we know that $\sum_{t=0}^{T} p(t)|\bra{L+1}e^{-iJt}\ket{1}|^{2} = \Omega\big( \frac{1}{L+2} \big)$. For $T = O\big( (L+2)^2\: \log(2(L+2)) \big)$, we conclude that there must be a $t = O(L^2 \log L)$ for which $|\bra{L+1}e^{-iJt}\ket{1}|^{2} = \Omega(1/L)$. 
\end{proof}

\section{Classically estimating entries of the time-evolved correlation matrix on lattice models}
\label{sec:classicalsim}
In this appendix we briefly argue the following. For $t = {\rm poly}(n)$ and assuming classical access to entries $\bra{k}M\ket{l}$ of an initial correlation matrix $M$ for given $(k,l)$, one can obtain entries $M(t)_{ij}$ with $1/\exp(n)$ error with poly$(n)$ classical effort. To see this, note that $\max_{x\in [-1,+1]}|p_K(x) - \exp(itsx)| = O\big( (t/\sqrt{K})^{K+1} \big)$, with $p_K(x)$ a degree-$K$ Taylor approximation. This implies 
\vspace{-0.3cm}
\begin{multline}
    \Bigl\lvert \big(p_{K}(h/s) M p_{K}(-h/s)\big)_{ij} - M(t)_{ij} \Bigr\rvert = \\ O\big( (t/\sqrt{K})^{K+1} \big),
\end{multline}
where we have used that $\|M\|\leq 1$. Note that this error can be bounded by $1/\exp(n)$ for some $K = {\rm poly}(n)$. Using the same reasoning as in the proof of Lemma \ref{lem:classicalsim}, we can obtain $\bra{i}h^{k_1} M h^{k_2}\ket{j}$ for all $k_1,k_2\leq K = {\rm poly}(n)$, giving an estimate of $\big(p_{K}(h/s) M p_{K}(-h/s)\big)_{ij}$. So for sufficiently large $K = {\rm poly}(n)$, we obtain an estimate of $\big(e^{+ith}M_{0}e^{-ith}\big)_{ij}$ with $1/\exp(n)$ error. 

\end{document}